\documentclass[acmsmall, screen]{acmart}

\AtBeginDocument{%
  }

\setcopyright{cc}
\setcctype{by}
\acmDOI{10.1145/3808346}
\acmYear{2026}
\acmJournal{PACMPL}
\acmVolume{10}
\acmNumber{PLDI}
\acmArticle{268}
\acmMonth{6}
\acmSubmissionID{pldi26main-p887-p}
\received{2025-11-14}
\received[accepted]{2026-04-03}

\usepackage{algorithm}
\usepackage{algpseudocode}
\usepackage{amsthm}
\usepackage{caption}
\usepackage{enumitem}
\usepackage{listings}
\usepackage{makecell}
\usepackage[most]{tcolorbox}
\usepackage{subcaption}
\usepackage{wrapfig}
\usepackage{xcolor}
\usepackage{xspace}

\definecolor{backcolour}{RGB}{248, 248, 248}

\usepackage[most]{tcolorbox}
\tcbset{
  highlightstyle/.style={
    colback=gray!10,
    colframe=gray!30,
    boxrule=0.5pt,
    arc=2pt,
    left=5pt,
    right=5pt,
    top=1pt,
    bottom=1pt,
  }
}

\newtcolorbox{grayboxtext}{highlightstyle}

\definecolor{backcolour}{rgb}{0.95,0.95,0.95}

\newcommand{\pheader}[1]{%
  \vspace{0.8ex}%
  \textbf{#1.}\ %
}

\newcommand{\toolname}{\textsc{AbsEvolve}\xspace}

\newtheorem{definition}{Definition}[section]
\newtheorem{theorem}{Theorem}[section]
\newtheorem{example}{Example}[section]

\newcommand{\block}{\mathcal{B}}

\newcommand{\algps}{UPOSE\xspace}
\newcommand{\alggs}{AGG\xspace}

\newcommand{\adom}{\mathcal{A}}
\newcommand{\cdom}{\mathcal{C}}
\newcommand{\conval}{\mathtt{c}}
\newcommand{\elina}{\textsc{Elina}\xspace}
\newcommand{\clam}{\textsc{Clam}\xspace}
\newcommand{\vars}{\mathcal{V}}
\newcommand{\op}{\mathcal{O}}
\newcommand{\exps}{\mathbb{E}}
\newcommand{\real}{\mathbb{R}}
\newcommand{\sigtwo}{\sigma_{\le 2}}
\newcommand{\guard}{\mathcal{G}}
\newcommand{\temp}{\mathcal{T}}
\newcommand{\tempspace}{\mathbb{R}^{t \times n}}
\newcommand{\absval}{\mathtt{a}}
\newcommand{\absvalin}{\mathtt{a}_{\mathit{in}}}
\newcommand{\absvalout}{\mathtt{a}_{\mathit{out}}}
\newcommand{\absvalf}[1]{\mathtt{a}_{#1}}
\newcommand{\psm}{\mathcal{M}}
\newcommand{\param}{\theta}
\newcommand{\parspace}{\Theta}
\newcommand{\lf}{L}
\newcommand{\eff}{\hat{F}}
\newcommand{\efc}{\hat{G}}
\newcommand{\seq}{\hat{S}}
\newcommand{\invs}{\mathcal{I}}

\newcommand{\elemspace}{\mathcal{S}_\temp}

\begin{document}

\title{Evolving Abstract Transformers for Gradient-Guided, Adaptable Abstract Interpretation}

\author{Shaurya Gomber}
\orcid{0009-0002-1783-4899}
\affiliation{%
  \institution{University of Illinois Urbana-Champaign}
  \city{Urbana}
  \country{USA}
}
\email{sgomber2@illinois.edu}

\author{Debangshu Banerjee}
\orcid{0009-0001-0163-9717}
\affiliation{%
  \institution{University of Illinois Urbana-Champaign}
  \city{Urbana}
  \country{USA}
}
\email{db21@illinois.edu}

\author{Gagandeep Singh}
\orcid{0000-0002-9299-2961}
\affiliation{%
  \institution{University of Illinois Urbana-Champaign}
  \city{Urbana}
  \country{USA}
}
\email{ggnds@illinois.edu}


\begin{abstract}
Current numerical abstract interpretation relies on fixed, hand-crafted, instruction-specific transformers tailored to each domain, giving rise to three significant limitations. First, extensibility is limited because transformers cannot be reused across domains and new transformers need to be designed for each new domain or operator. Second, precise compositional reasoning over instruction sequences is difficult as transformers are defined only at the instruction level. Third, all downstream tasks are forced to use the same fixed transformer, irrespective of their precision, efficiency, or task-specific requirements. To address these limitations, we propose the \textit{Evolving Abstract Transformer}, a general transformer that replaces the fixed single-output design of traditional transformers with an adaptable search over a parametric space of sound outputs. This is achieved through two underlying algorithms we develop. First, the \textit{Universal Parametric Output Space Encoder (UPOSE)} constructs a compact parametric space of sound outputs for any polyhedral numerical domain and any operator in the \textit{Quadratic-Bounded Guarded Operators (QGO)} class which includes both individual instructions and structured sequences. Next, the \textit{Adaptive Gradient Guidance (AGG)} algorithm leverages the differentiable structure of the space generated by \algps and uses gradient-based updates to efficiently search it according to downstream analysis objectives and available runtime, continually \textit{evolving} the output as more time is provided. We implement these ideas in the \toolname framework and evaluate their effectiveness across three numerical abstract domains: Zones, Octagons, and Polyhedra. Our results demonstrate that the evolving transformer works across domains and handles diverse instructions and sequences, allowing efficient adaptability in the precision–efficiency tradeoff by adjusting the number of gradient steps in the search, while also reaching the most precise invariants up to 3.2× faster than existing baselines.
\end{abstract}




\maketitle
\section{Introduction}
Abstract interpretation~\cite{absinterp} is a widely used framework for static program analysis that soundly approximates a program’s \textit{concrete} semantics using an alternate representation called the \textit{abstract domain}. Over the years, a variety of domains have been developed to capture different properties, such as heap structure~\cite{heap1, checker}, control-flow behavior~\cite{cntrlFlow}, and numeric relationships~\cite{polyhedra, subPoly}. This work focuses on \textit{numerical} domains, which are central to reasoning about variable ranges~\cite{SinghPV17}, arithmetic updates~\cite{subPoly}, and control-flow conditions~\cite{cntrlFlow} in real-world programs. Several numerical domains exist, such as Interval~\cite{interval}, Octagons~\cite{octagon} and Polyhedra~\cite{polyhedra}, each varying in expressiveness and efficiency. Selecting the right domain involves the fundamental trade-off between precision and scalability. Once the domain is chosen, the analysis is performed through the domain's \textit{abstract transformers}, which are functions that define how program operations (e.g., assignments) are soundly over-approximated within the domain.

\pheader{Challenge 1: Per-Domain and Per-Instruction Design}
Designing an abstract transformer is complicated, since it requires balancing efficiency and precision while ensuring soundness throughout. This must be carried out for all possible instructions across all abstract domains. Moreover, precision in abstract interpretation is \textit{non-compositional}, i.e., analyzing an instruction sequence per-instruction and composing the results is \textit{generally} less precise than analyzing the sequence jointly. Thus, in principle, achieving precise analysis would require designing separate abstract transformers for all instruction sequences. Designing such a large number of transformers manually is very tedious and has motivated efforts toward automating transformer construction~\cite{autoabs}. Prior work such as~\cite{abstractionSyn1, abstractionSyn2} uses program synthesis to generate transformers within a user-defined DSL for specific domains such as strings. However, such automation is not possible for general numerical domains like Octagons and Polyhedra, which involve reasoning about numerical relations that are more complex to handle. Consequently, current numerical analysis relies on libraries such as ELINA~\cite{SinghPV17}, APRON~\cite{apron}, and PPL~\cite{ppl}, which support limited standard domains such as Octagons and Polyhedra by providing hand-crafted transformers for them. These libraries face two key limitations. First, many analysis scenarios require domains beyond the standard ones, such as TVPI~\cite{tvpi}, Octahedron~\cite{octah}, Logahedra~\cite{logh} etc. Although these domains resemble existing ones like Octagons in being \textit{polyhedral} and representing states via conjunctions of linear constraints, the abstract transformers must be redesigned for each domain, as the existing transformers are tied to their specific domains and cannot be reused. Second, because the transformers are manually designed, it is infeasible to construct separate transformers for all instruction sequences, so they are designed only at the instruction level. As a result, instruction sequences are analyzed by composing per-instruction transformers, which reduces precision.

\pheader{Challenge 2: Imprecision \& Lack of Adaptability}
The loss of precision from composing transformers is compounded by the fact that current libraries often do not compute the most precise transformer even for simple operations. For instance, Template Constraint Matrix (TCM) domains~\cite{tcm_domain} such as Octagons are widely used, as they are more efficient than expressive domains like Polyhedra, whose operations have exponential complexity. However, these domains cannot precisely capture even simple affine assignments such as $z = 2x + 3y$ that fall outside their constraint form. Current libraries worsen this imprecision by not even computing the most precise over-approximation within the domain, as it is computationally expensive and involves solving a large number of linear programs per transformer call. Moreover, domains like Polyhedra can represent linear assignments exactly but face similar limitations for non-linear assignments like $z = x * x$, whose precise analysis requires calls to non-linear solvers that are extremely slow and do not even guarantee soundness~\cite{gurobi_nonlinear_unsound}. To remain efficient, existing libraries use
ad hoc handcrafted over-approximations that are sound but often highly imprecise.

Furthermore, current libraries rely on \textit{fixed} transformers that always prioritize efficiency over precision, \textit{irrespective of the analysis task}. Different scenarios, however, demand different trade-offs: lightweight runtime checkers~\cite{checker} require fast analyses, whereas safety-critical analyzers~\cite{astree} can afford longer runtimes for higher precision. 
Prior works~\cite{tune1,tune2} have shown the benefits of adapting analyzers to downstream goals by tuning general parameters that control their behavior, while~\cite{bai,resourceaware} highlight the need and benefits of adapting analyses to time and memory constraints. Current abstract transformers, however, are not amenable to such tuning, as these transformers are rigid and always follow the same routine, producing a single fixed output (from among potentially infinite sound outputs) that cannot be tuned for precision or compute efficiency.

\pheader{Key Idea: Evolving Abstract Transformer}
To overcome the challenges discussed, we introduce \emph{Evolving Abstract Transformer}, a new paradigm that unifies transformer design across various domains and instructions and replaces the fixed output design with an adaptable, search-based process.
The design proceeds in two key phases. For a given input abstract element, our \textit{Universal Parametric Output Space Encoder} (\algps) constructs a \textit{symbolic, parametric} space of sound outputs, rather than computing a single fixed output. The \textit{Adaptive Gradient Guidance} (\alggs) algorithm then starts from an initial point in this space and efficiently traverses it through gradient-based updates, iteratively refining and \emph{evolving} the outputs toward increasingly precise and always sound results. This novel search-based design enables analyses to adapt precision over time and converge toward more precise outputs within the available time budget. This design offers the following benefits:

\begin{enumerate}[leftmargin=*]

\item \textbf{Cross-Domain and Sequence Generality:}
Unlike existing transformers, our evolving transformer is not tied to a specific domain or instruction. The \algps algorithm works for all Polyhedral domains (Section~\ref{sec:background}), from standard ones like Octagons, Polyhedra to domains like Octahedron, and handles a broad range of instructions, including affine, quadratic, and guarded assignments captured by the Quadratic-Bounded Guarded Operators (QGO) class (Section~\ref{sec:upose}). The QGO class is highly expressive and also captures instruction sequences whose combined semantics lie in this class (Sec.~\ref{sec:framework}). The \algps algorithm enables efficient abstraction of such sequences as a whole, outperforming SMT-based symbolic abstraction approaches such as~\cite{awsSolver}.

\item \textbf{Parametric Space for Adaptable Analysis}: Instead of computing a single sound output, the \algps algorithm synthesizes a \emph{sound-by-construction}, \emph{compact parametric} space of outputs, encoded using a small set of parameters. This design makes our evolving transformer flexible, allowing analyses to adaptively search within an always-sound output space to find results that best suit their needs.
For instance, the space of sound outputs synthesized by \algps includes the precise outputs computed by solver-based transformers~\cite{tcm_domain,awsSolver}, enabling analyses to achieve solver-level precision when such outputs are discovered by the search.

\item \textbf{Efficient Gradient-Guided Search:}
The effectiveness of our evolving transformer depends on the efficiency of its underlying search. This is addressed by the \alggs algorithm, which exploits the fact that the parametric space produced by \algps is \textit{differentiable}, and uses gradient information to efficiently guide the search and refine outputs within a given \emph{runtime budget}~$\mathcal{R}$ (number of gradient steps). This gradient-guided process enables continual refinement and efficient use of runtime toward achieving higher precision, allowing \alggs to converge to the most precise outputs much faster than solver-based transformers~\cite{tcm_domain,awsSolver}, as shown in Section~\ref{sec:eval1}. More generally, \alggs can optimize any search objective~$\mathcal{J}$ (score function) using its gradient~$\nabla\mathcal{J}$, extending beyond precision to other analysis goals.

\end{enumerate}

\begin{figure}[!t]
    \includegraphics[width=1\linewidth]{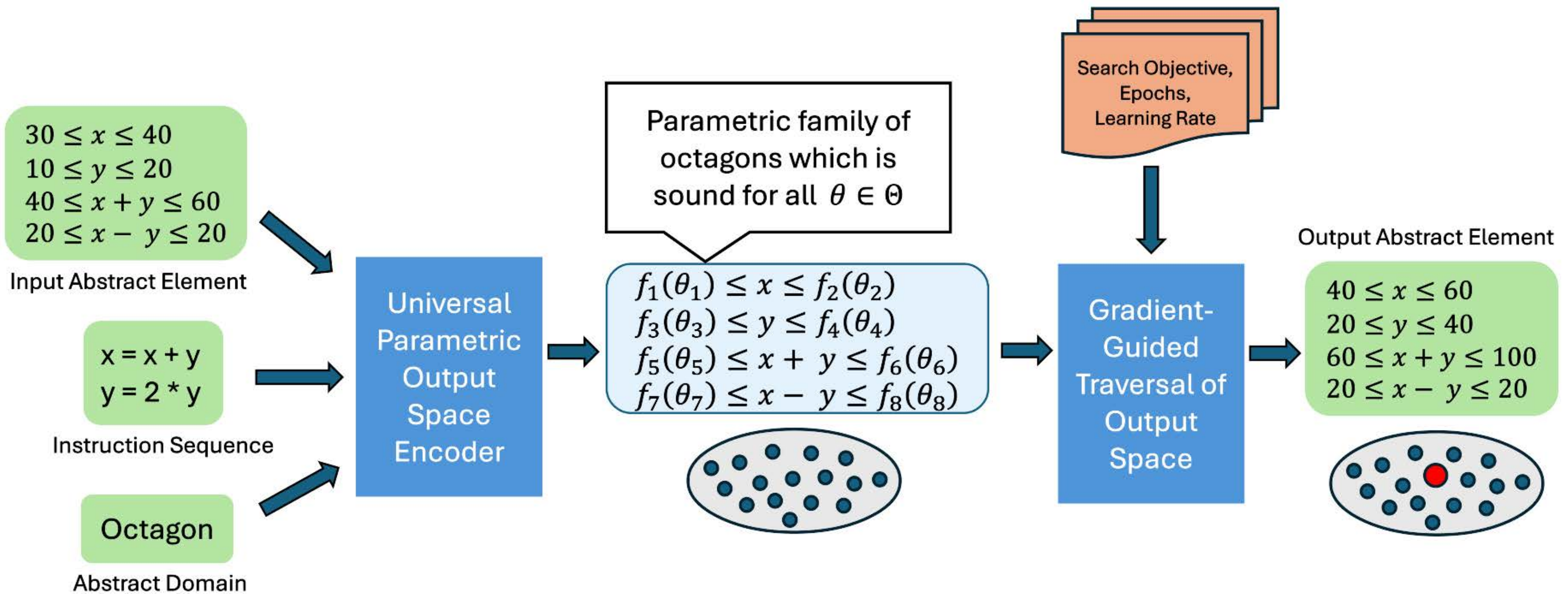}
\caption{For a given numerical domain, instruction sequence, and input abstract element, \toolname first generates a parametric space of sound outputs for that sequence. Then, it performs gradient-guided search, driven by user-specified parameters, to efficiently select a suitable output during analysis.}
\label{fig:overview-fig}
\vspace{-2pt}
\end{figure}

\textbf{\toolname framework.}
We implement these ideas in the \toolname framework (Fig.~\ref{fig:overview-fig}). \toolname performs abstract interpretation using Evolving Abstract Transformers, applying the \algps algorithm to encode a parametric space of sound outputs for the supported operators and the \alggs algorithm to select suitable outputs during analysis. For instruction sequences whose combined effect lies within the QGO class, \toolname automatically merges them and uses \algps to encode a joint parametric space for the entire block (Sec~\ref{sec:framework}), reducing precision loss from composing per-instruction transformers. For unsupported operations, the framework falls back to baseline transformers (from \elina). Users can specify parameters such as the search budget~$\mathcal{R}$ (number of gradient steps) to control the search process and adapt precision based on available resources.

\begin{grayboxtext}
\toolname is the first precision–efficiency trade-off adaptable, sound-by-construction framework for abstract interpretation that combines formal soundness guarantees with the efficiency of gradient-guided optimization.
\end{grayboxtext}

\pheader{Main Contributions} This paper makes the following contributions:
\begin{itemize}[leftmargin=*]
    \item We introduce the \emph{Evolving Abstract Transformer}, a new general transformer that rethinks the traditional design by constructing a parametric space of sound outputs rather than a single output. This is achieved by using our Universal Parametric Output Space Encoder (\algps) algorithm, which operates across polyhedral numerical domains and handles a wide range of instructions and sequences from the Quadratic-Bounded Guarded Operators (QGO) class (Section~\ref{sec:upose}).

    \item To explore the space of outputs efficiently, we propose the novel Adaptive Gradient Guidance (\alggs) algorithm, which takes a search objective $\mathcal{J}$ and runtime budget $\mathcal{R}$ as input, and efficiently navigates the output space using gradient signals from $\mathcal{J}$ to guide the search toward outputs that best align with the requirements of the downstream tasks (Section~\ref{sec:agg}).

    \item We implement these algorithms in the \toolname framework (Section~\ref{sec:framework}) and evaluate them across three numerical domains: Zones, Octagons, and Polyhedra. Our results (Section~\ref{sec:eval}) show that \toolname enables compositional reasoning over instruction sequences and achieves significant, adaptable precision improvements over baselines through gradient-guided evolution, while also reaching the most precise invariants up to 3.2× faster than existing baselines.
\end{itemize}

\section{Background}
\label{sec:background}

\pheader{Abstract Interpretation}
Abstract interpretation~\cite{absinterp} over-approximates the behavior of a system represented by a \textit{Concrete Domain} using an alternate \textit{Abstract Domain}. The concrete domain $(\cdom, \sqsubseteq_\cdom)$ consists of concrete elements partially ordered by~$\sqsubseteq_C$, while the abstract domain $(\adom, \sqsubseteq_\adom)$ contains abstract elements used to \textit{represent} elements in $\cdom$. In program analysis, a \emph{state} is a valuation over program variables~$\vars$, and the concrete domain $\cdom$ is the powerset of all program states, $\mathcal{P}(\mathit{States})$, ordered by subset inclusion~$\subseteq$. Each element $\conval \in \cdom$ thus represents a set of possible program states. Abstract domains such as Intervals, Zones, or Polyhedra provide finite representations of such sets, enabling efficient analysis by over-approximating elements of~$\cdom$. For the rest of this paper, we assume the program analysis setting with $\cdom = \mathcal{P}(\mathit{States})$.

\pheader{Concretization Function}
The \emph{concretization} function $\gamma : \adom \to \cdom$ is a monotonic function, which, for a given abstract element $\absval \in \adom$ computes the concrete element $\gamma(\absval) \in \cdom$ represented by $\absval$. $\gamma(\absval)$ captures all program states represented by $\absval$.

\begin{definition}[Precision of Abstract Elements]
\label{def:precision}
An abstract element $\absvalf{1} \in \adom$ is \emph{more precise} than another element~$\absvalf{2} \in \adom$ if $\gamma(\absvalf{1}) \subseteq \gamma(\absvalf{2})$.
\label{def:prec-abs}
\end{definition}

\textbf{Abstract Transformers.} The concrete domain includes operations $op : \cdom \to \cdom$ that operate on concrete elements. In program analysis, these operations correspond to statements such as assignments that transform program states. To analyze the concrete domain within the abstract domain, we require, for each concrete operator~$op$, a corresponding abstract operator~$\hat{op} : \adom \to \adom$ that operates on abstract elements in a way that \textit{soundly} reflects the effect of~$op$ on concrete elements. This function~$\hat{op}$ is called the \emph{abstract transformer} corresponding to the concrete operator~$op$. Now, we define the soundness of abstract transformers, followed by the notion of precision between transformers and the definition of the most-precise transformer.

\begin{definition}[Soundness of Abstract Transformer]
The abstract transformer $\hat{op} : \adom \to \adom$ for concrete operator $op$ is sound if it over-approximates the effect of $op$, i.e. for all abstract elements $\absval \in A$, the inclusion relation $op(\gamma(\absval)) \subseteq \gamma(\hat{op}(\absval))$ holds.
\end{definition}

\begin{definition}[Precision of Abstract Transformers]
\label{def:precision-transformer}
Given two sound abstract transformers~$\hat{op}_1, \hat{op}_2 : \adom \to \adom$ for the same concrete operator~$op$, we say that~$\hat{op}_1$ is \emph{more precise} than~$\hat{op}_2$ if it generates more precise outputs, i,e, $\forall \absval \in \adom, \; \gamma(\hat{op}_1(\absval)) \subseteq \gamma(\hat{op}_2(\absval))$.
\end{definition}

\begin{definition}[Most-Precise Abstract Transformer]
\label{def:best-transformer}
An abstract transformer~$\hat{op}^{\#} : \adom \to \adom$ is the \emph{most-precise} (or \emph{best}) transformer for a concrete operator~$op : \cdom \to \cdom$ in domain $\adom$ if it is sound and generates most-precise outputs, i.e, for every other sound abstract transformer~$\hat{op}' : \adom \to \adom$ corresponding to~$op$,  $\forall \absval \in \adom, \; \gamma(\hat{op}^{\#}(\absval)) \subseteq \gamma(\hat{op}'(\absval))$ holds.
\end{definition}

\textbf{Polyhedral Domains.}
Polyhedral domains are a widely used family of numerical domains that represent sets of program states using systems of linear inequalities. Given a finite set of program variables~$\vars = \{v_1, \ldots, v_n\}$, an abstract element is~$\absval = \{\mathbf{A}\mathbf{v} \le \mathbf{b}\}$, where~$\mathbf{v} = [v_1, \ldots, v_n]^\top$, $\mathbf{A} \in \mathbb{R}^{m \times n}$, and~$\mathbf{b} \in \mathbb{R}^m$. This formulation subsumes a range of widely used domains: Intervals~\cite{interval} capture independent bounds on variables, Zones~\cite{zones} track differences between variables, Octagons~\cite{octagon} handle sums and differences, and general Polyhedra~\cite{polyhedra} support arbitrary linear relations.

\pheader{Template Constraint Matrix (TCM) Domains}
A common subclass of polyhedral domains is the family of \emph{template-based domains}, which restricts abstract elements to a fixed set of linear directions. These directions are encoded by a \emph{Template Constraint Matrix (TCM)}~\cite{tcm_domain} $\temp \in \tempspace$, where each row~$\temp_i$ specifies a linear constraint over the finite set of variables~$\mathcal{V} = \{v_1, \ldots, v_n\}$. Letting $\mathbf{v} = [v_1, \ldots, v_n]^\top$, abstract elements take the form $\absval = \{\mathcal{T} \mathbf{v} \ge \mathbf{c}\}$, where $\mathbf{c} \in (\mathbb{R} \cup \{-\infty\})^t$ gives lower bounds for each constraint. Any constraint with $c_i = -\infty$ is treated as vacuously true and omitted. TCM domains include domains like Interval, Zones, Octagons, etc.
\section{Overview}

\begin{figure}
    \centering
    \begin{subfigure}[t]{0.48\linewidth}
        \centering
\begin{lstlisting}
int x = 0, y = 0;

while (y <= 10)
{
   x = x + y;
   y = y + 1;
   x = x - y;
}

assert(x <= 0);
\end{lstlisting}
        \caption{Imprecision due to imprecise assignment ~\\transformers}
        \label{lst:octa}
    \end{subfigure}
    \hfill
    \begin{subfigure}[t]{0.48\linewidth}
        \centering
\begin{lstlisting}
int x = 30;
int y = 10;

while (y <= 80)
{
   x = x + y;
   y = 2*y;
}

assert(x - y <= 20);
\end{lstlisting}
        \caption{Imprecision due to statement-wise abstract ~\\interpretation}
        \label{lst:octb}
    \end{subfigure}
    \caption{Examples demonstrating imprecision of current libraries for octagon analysis.}
    \label{fig:oct-codes}
\end{figure}
In this section, we highlight the key limitations of the current abstract interpretation design and demonstrate how our \textit{Evolving Abstract Transformer} helps overcome them.

\subsection{Limitations of Current Design}
Consider the program in Fig.~\ref{lst:octa}. Since $x = 0$ at start of the loop, $x \le 0$ holds initially. In each iteration, $x$ increases by $y$ and then decreases by $y + 1$, effectively decreasing by 1. So, $x \le 0$ holds always and is a valid loop invariant, implying that the assertion holds. However, when analyzed with the state-of-the-art \elina~\cite{SinghPV17} library using the Octagon domain, the inferred invariant does not include $x \le 0$ and assertion remains unproven, even though this constraint lies in the octagonal template. This is because the abstract transformers used for assignments in \elina are \textit{imprecise}. For instance, at Line~7, the following transition arises during the analysis:
\begin{equation}
o_{\text{in}} = \{\, x \le 10,\; 1 \le y \le 11,\; x - y \le -1,\; x + y \le 21 \}\
\quad
x := x - y
\quad
\{\, o_{\text{out}} \,\}
  \label{eq:overview0}
\end{equation}

In the Octagon domain, each constraint has the form $\pm x \pm y \ge c$, and computing the most precise $o_{\text{out}}$ requires solving one linear program per constraint direction. To compute the tightest lower bound for $-x$, we minimize $y - x$ (as new $x$ value is $x - y$) under the current constraints:

\begin{equation}
  c^\# = \min_{x, y} (y - x) \quad \text{s.t.} \quad x \le 10,\; 1 \le y \le 11,\; x - y \le -1,\; x + y \le 21
  \label{eq:overview1}
\end{equation}

\pheader{Limitation 1: Fixed Imprecision}
Solving the above LP yields $c^\# = 1$, so the most precise bound for $-x$ is $-x \ge 1$. As solving such optimization problems for all template directions (quadratic in the number of variables for Octagons) at each program point is computationally expensive, \elina uses a sound but imprecise \emph{interval relaxation} heuristic that ignores relational constraints and reasons only over variable intervals. Given $x \in [-\infty, 10]$ and $y \in [1, 11]$, this yields $y - x \in [-9, \infty]$ (and thus $c^\# = -9$), giving $-x \ge -9$, which is sound but much less precise than the true bound $-x \ge 1$. This loss of precision ultimately prevents the analyzer from discovering the invariant $x \le 0$.

\pheader{Limitation 2: Lack of Adaptability}
The key limitation is not just imprecision, but also rigidity. Although the most precise LP-based transformer could prove the assertion, \elina never computes it and always defaults to the interval heuristic for efficiency, \textit{regardless of analysis needs}. This fixed design lacks adaptability; for instance, in safety-critical settings where higher precision is crucial, transformers in the current design would offer no way to trade efficiency for precision.

\pheader{Limitation 3: Imprecision from Per-instruction design}
Libraries like \elina have hard-coded transformers for all instructions, such as assignments, across every domain they support. During analysis, these transformers are applied instruction-by-instruction. It is well known that precision in abstract interpretation is non-compositional, and that this instruction-wise analysis introduces imprecision, as relational information between variables is lost when dependent instructions are analyzed separately. Consider the program in Fig.~\ref{lst:octb}. The relation $x - y = 20$ holds initially. Since both $x$ and $y$ are updated using the current value of $y$ ($2y = y + y$), $x - y$ remains constant and $x - y = 20$ is a valid loop invariant, so the assertion holds. However, neither current libraries nor substituting their transformers with LP-based ones can prove it, as the core issue is that both assignments must be analyzed together. This limitation arises because current transformers are implemented separately for each domain and instruction. In such setup, it is infeasible to design and maintain transformers for all possible instruction sequences that may arise during analysis.

\pheader{This Work: Evolving Abstract Transformer}
To address these limitations, we move beyond fixed, hard-coded transformers and introduce the \emph{Evolving Abstract Transformer}, a novel general transformer that reformulates the traditional approach of computing a single sound output into one that defines a parametric space of sound outputs and searches within it. This enables the transformer to adapt and \emph{evolve} toward higher-precision outputs through guided optimization. This is enabled by two core algorithms we propose, the \textit{Universal Parametric Output Space Encoder} (\algps) and the \textit{Adaptive Gradient Guidance} (\alggs), which we describe next.

\subsection{Universal Parametric Output Space Encoder (\algps) Algorithm}
While analyzing the program in Fig.~\ref{lst:octa} above, we discussed two sound lower bounds for $-x$: (1)~1, the \textit{most precise} bound obtained by minimizing $y - x$ under the current constraints (Eq~\ref{eq:overview1}), and (2)~$-9$, a much looser bound produced by the interval relaxation heuristic applied by \elina. These are not the only sound lower bounds: any value below $1$ is sound, with bounds closer to $1$ being more precise but costlier to compute (often requiring LP solving), while faster heuristics like interval relaxation yield less precise results. Since analysis tasks have varying precision and time requirements, our \textit{UPOSE} algorithm does not commit to one output but instead captures the space of sound bounds \textit{symbolically}, enabling the subsequent search procedure to adapt precision based on the available time budget. To capture this space of sound bounds, \algps exploits the notion of duality and uses the \textit{Lagrangian dual} function. For a problem of the form $\min_x f(x)\ \text{s.t.}\ Ax \le b$, the Lagrangian dual function is
$g(\boldsymbol{\lambda}) = \min_x \bigl(f(x) + \boldsymbol{\lambda}^\top(Ax-b)\bigr)$. By the weak duality theorem~\cite{boyd2004convex}, the dual function of the minimization problem in Eq.~\ref{eq:overview1} satisfies $g(\boldsymbol{\lambda}) \le c^\#$ for every feasible $\boldsymbol{\lambda}\!\ge\!0$, giving us a principled way to parametrize the space of lower bounds $g(\boldsymbol{\lambda})$ in terms of Lagrange multipliers $\boldsymbol{\lambda}$. However, using this form directly has two challenges: (i) the inner minimization defining $g(\boldsymbol{\lambda})$ is not trivial to compute efficiently, and (ii) many choices of multipliers $\boldsymbol{\lambda}$ make this minimization diverge to $-\infty$, yielding trivial meaningless bounds. \algps overcomes these challenges and constructs, for each constraint direction, a \textit{Parametric Scalar Map (PSM)} $\psm = \{\parspace, \lf\}$, where each $\param \in \parspace$ defines a \textit{finite} sound lower bound $\lf(\param)$. 

First, to efficiently compute and bound $g(\boldsymbol{\lambda})$, \algps creates the dual function $g(\boldsymbol{\lambda})$ by introducing Lagrange multipliers \emph{only for the relational} constraints (which capture inter-variable dependencies) and keeping variable bounds intact. The key motivation is that variable bounds define a bounding box over which the inner minimization can be solved cheaply in closed form, as we will see later. Introducing multipliers for all constraints would eliminate this structure, making the inner minimization hard to compute. Instead, we retain the box constraints and only dualize the relational constraints such as $\pm x \pm y \leq c$ in octagons. Most numerical domains already track sound variable bounds, and \algps also works even when some or all of these bounds are absent. For our running example (Eq.~\ref{eq:overview1}), introducing non-negative multipliers $\boldsymbol{\lambda} = (\lambda_1, \lambda_2)$ for the relational constraints
$x - y \le -1$ and $x + y \le 21$ yields:
\[
    o_{dual}
    = \max_{\lambda_1, \lambda_2 \ge 0}
      \;\; \min_{x \in [-\infty, 10],\, y \in [1,11]}
      \; y - x
      + \lambda_1(x - y + 1)
      + \lambda_2(x + y - 21)
\]

After simplifying the inner minimization, this construction gives the dual function:
\[
g(\lambda_1, \lambda_2)
    = \min_{x \in [-\infty, 10],\; y \in [1,11]}
      (\lambda_1 + \lambda_2 - 1)x
    + (-\lambda_1 + \lambda_2 + 1)y
    + \lambda_1 - 21\lambda_2
\]

Now, evaluating $g(\lambda_1, \lambda_2)$ over feasible $\lambda_1, \lambda_2 \ge 0$ yields a continuous family of sound bounds parameterized by~$\boldsymbol{\lambda}$. But as mentioned above, certain choices of $\lambda_1$ and $\lambda_2$ make the inner minimization diverge to~$-\infty$, which is a sound but trivial lower bound. Hence, while computing~$g$, \algps also derives additional constraints on the multipliers to ensure that the resulting parametric space contains only sound and \textit{finite} lower bounds meaningful for downstream analysis. To compute $g(\lambda_1,\lambda_2)$, \algps uses the fact that the retained variable bounds make the inner minimization over an interval hyperbox. This structure allows a closed-form minimum (which would not be possible without the bounds) by treating each variable separately and computing its contribution from its bounds. We start with $\psm = \{\parspace, \lf\}$, where $\parspace$ initially includes $\lambda_1, \lambda_2 \ge 0$ and $\lf = 0$, and we iteratively update it with additional constraints while 
accumulating each variable’s contribution to $L$ as follows:
\begin{enumerate}[leftmargin=*]
    \item The $x$ term $(\lambda_1+\lambda_2-1)x$ is minimized over $x \in [-\infty,10]$. Since the lower bound is unbounded, finiteness requires $(\lambda_1+\lambda_2-1) \le 0$, preventing the minimum from diverging. We add this constraint to $\parspace$ and the contribution $\lf_x = 10(\lambda_1+\lambda_2-1)$ to $\lf$.
    \item The $y$ term $(-\lambda_1+\lambda_2+1)y$ is minimized over $y \in [1,11]$, where both bounds are finite, so the result is finite and $\parspace$ remains unchanged. The minimum depends on the sign of the coefficient: if $(-\lambda_1+\lambda_2+1) \ge 0$, it occurs at $y=1$; otherwise, at $y=11$. It is easy to verify that these cases combine into the closed form $L_y = -5|\lambda_2-\lambda_1+1| + 6(\lambda_2-\lambda_1+1)$ which we add to $\lf$.
\end{enumerate}
Adding the constant term $(\lambda_1 - 21\lambda_2)$ to $\lf$ yields the final scalar map $\psm = \{\parspace, \lf\}$, where, after simplification, $\lf = 5\lambda_1 - 5\lambda_2 - 5|\lambda_2 - \lambda_1 + 1| - 4$ and $\parspace = \{\lambda_1 \ge 0,\ \lambda_2 \ge 0,\ \lambda_1 + \lambda_2 - 1 \le 0\}$. This captures the space of sound lower bounds for our problem (Eq.~\ref{eq:overview1}). Collecting such maps for all template directions yields the \textit{parametric space of sound outputs} (details in Sec.~\ref{sec:upose}). As discussed in Sec.~\ref{sec:upose}, the constructed output space has the property that setting all parameters (multipliers) to zero recovers the interval relaxation result. In this example, $\boldsymbol{\lambda}=(0,0)$ satisfies $\parspace$ and gives the finite lower bound $L=-9$, matching the interval bound discussed earlier. The parametric space also contains the most precise lower bound for linear operators (like the linear assignment in our example). Choosing $\boldsymbol{\lambda}^* = (1,0)$ yields $L=1$, matching the most precise bound discussed above.

\pheader{Construction for quadratic case}
The example discussed above demonstrates the construction for a linear operator. However, the \algps algorithm also supports quadratic operators. For instance, if the operator at Line~7 is replaced by the quadratic update $x := -x^2$, while keeping the same input bounds and relational constraints, the corresponding inner minimization becomes:
\[
\min_{x \in [-\infty,10],\; y \in [1,11]}
\; x^2 + (\lambda_1+\lambda_2)\,x + (-\lambda_1+\lambda_2)\,y + \lambda_1 - 21\lambda_2,
\quad \lambda_1,\lambda_2 \ge 0.
\]
This minimization is harder than in the linear case because the objective mixes a quadratic term $x^2$ with a linear term in~$x$, making it unclear how to separate their contributions as before. To address this, \algps introduces a split parameter~$s$ that decomposes the $x$-coefficient as $(\lambda_1+\lambda_2) = s + (\lambda_1+\lambda_2 - s)$, allowing the objective to be separated into quadratic and linear parts, which yields:
\[
\hat{g}(\lambda_1,\lambda_2,s)
=
\underbrace{\min_{x}(x^2 + s x)}_{\text{quadratic part}}
+
\underbrace{\min_{x,y}[(\lambda_1+\lambda_2-s)x + (-\lambda_1+\lambda_2)y + \lambda_1 - 21\lambda_2]}_{\text{linear part}},
\]
We have $\hat{g}(\lambda_1,\lambda_2,s) \le g(\lambda_1,\lambda_2) \le c^\#$, where the first inequality holds because the sum of independent minima is always less than or equal to the joint minimum, and the second follows from weak duality. Hence, $\hat{g}(\lambda_1,\lambda_2,s)$ defines a sound parametric space of lower bounds for the quadratic operator. Linear part is minimized as discussed above.  For the quadratic part, similar conditions for finiteness and minimized value can be deduced by checking the interval boundaries and the critical point (details in Appendix~\ref{subsec:inner-parts-eval}).  Although strong duality does not hold for nonlinear operators and the space may not include the exact optimum~$c^\#$, this formulation enables \toolname to achieve more precise invariants for non-linear assignments in practice (Sec.~\ref{sec:eval2}).

\pheader{Supports Multiple Domains and Operator Sequences}
The \algps algorithm is general and works uniformly across all Polyhedral numerical domains (Sec.~\ref{sec:background}), and supports the broad range of \emph{Quadratic-Bounded Guarded Operators (QGOs)} (Sec.~\ref{sec:upose}).  QGOs combine linear or quadratic updates with guards expressed as conjunctions of linear inequalities, covering affine assignments, quadratic updates, guarded assignments, and conditionals like \texttt{assume}.  
This formulation extends naturally beyond single assignments, as many common instruction sequences have combined semantics that lie within the QGO class and can therefore be analyzed jointly.  
For instance, the sequence $[x := a + b;\; y := x \cdot c]$ yields the final updates $x = a + b$ and $y = (a + b) \cdot c$, both of which are quadratic-bounded, and the sequence lies in the QGO class. The \toolname framework (Sec.~\ref{sec:framework}) merges such sequences into blocks and uses the \algps algorithm to analyze them jointly and yield more precise invariants. This enables \toolname to infer the invariant $x - y = 20$ and prove the assertion in Fig.~\ref{lst:octb}, by analyzing the loop body $[x := x - y;\; y := 2 \cdot y]$ as a single unit.

\subsection{Adaptive Gradient Guidance (\alggs) Algorithm}
The parametric family of outputs computed by \algps is defined by $t$ Parametric Scalar Maps (PSMs), indexed by $i \in \{1, \ldots, t\}$. Each PSM $\psm_i = \{\parspace_i, \lf_i\}$ characterizes the space of sound, finite lower bounds admissible for the $i$-th constraint. The next step is to traverse this space and pick concrete bounds which can then be used for further analysis. We only need to consider non-empty $\parspace_i$ as if $\parspace_i$ is empty or infeasible, it indicates that no finite lower bound can be inferred for that direction, and the only sound bound under our construction is $-\infty$. For efficient traversal of the space, we propose the \textit{Adaptive Gradient Guidance (AGG)} (Sec.~\ref{sec:agg}) algorithm, which is a gradient-guided search procedure to select $\param_i \in \parspace_i$ suitable for analysis tasks. It takes as input a differentiable score function $\mathcal{J}_i$ that evaluates the quality of each bound $\lf_i(\param)$ for the analysis task via $\mathcal{J}_i(\lf_i(\param))$, as well as a number of gradient steps $\mathcal{R}$ that controls the runtime of the search. Exploiting the piecewise differentiability of the constructed space~$\lf$ (discussed in Sec~\ref{sec:upose}), AGG uses the gradient~$\nabla \mathcal{J}_i(\lf_i(\param))$ to efficiently navigate it, steering toward high-scoring bounds while keeping~$\param$ within the feasible region~$\parspace_i$. When optimizing for precision, where higher lower bounds indicate greater precision, the score function $\mathcal{J}_i$ can simply be set to $\lf_i$, and AGG performs gradient ascent to obtain tighter sound bounds within the feasible region. For instance, the PSM for our example above has  
$\lf(\lambda_1,\lambda_2) = 5\lambda_1 - 5\lambda_2 - 5|\lambda_2 - \lambda_1 + 1| - 4$  
with $\parspace = \{\lambda_1 \ge 0,\; \lambda_2 \ge 0,\; \lambda_1 + \lambda_2 \le 1\}$.  
To obtain more precise bounds (i.e., higher~$\lf$), the Adaptive Gradient Guidance (AGG) algorithm starts from the interval relaxation baseline $(\lambda_1,\lambda_2) = (0,0)$ and performs gradient-guided updates that increase~$\lf$ within~$\Theta$:  
\[
((\lambda_1,\lambda_2),\lf):\;
((0,0),-9)\; \Rightarrow\;
((0.3,0),-6)\; \Rightarrow\;
((0.6,0),-3)\; \Rightarrow\;
((1,0),1)\; \text{(most precise)}.
\]

\begin{figure}
    \centering
    \includegraphics[width=1\linewidth]{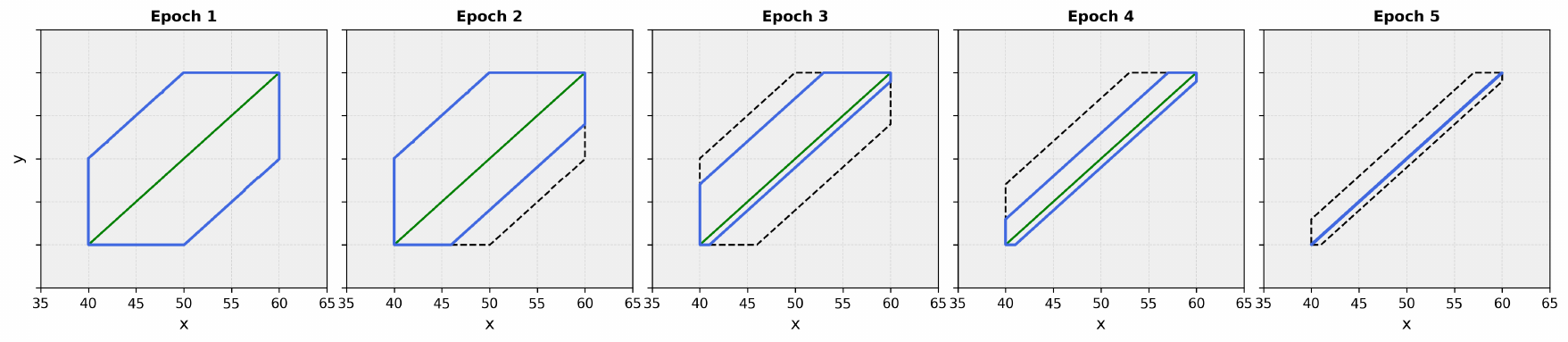}
\caption{Transformer output (blue) converging to most precise output (green) over successive gradient steps.}
    \label{fig:overview-epochs}
\end{figure}

The AGG procedure converges to the most precise bound in just three iterations. Figure~\ref{fig:overview-epochs} illustrates the per-step progression of AGG as it converges to the most precise invariant $x - y = 20$ during the analysis of program in Fig.~\ref{lst:octb}. The evaluation results in Section~\ref{sec:eval1} also demonstrate that AGG enables finding the most precise outputs significantly faster than solver-based baselines, highlighting the efficacy of gradient-guided search in efficiently exploring the parametric space.
\section{Construction of Parametric Space of Sound Outputs}
\label{sec:upose}

In this section, we formalize the notion of a parametric space of sound outputs and describe how the \emph{Universal Parametric Output Space Encoder (\algps)} algorithm systematically constructs such a space. The algorithm works for any operator whose overall effect can be expressed as a \emph{Quadratic-Bounded Guarded Operator (QGO)}. We begin by defining the notion of an \emph{Effective Update Map (EUM)}, which compactly represents the symbolic effect of an operator on program variables.

\begin{definition}[Effective Update Map]
Let $\vars$ be the finite set of program variables and $\op$ a concrete operator (atomic or composite). The \emph{Effective Update Map (EUM)} of~$\op$ is a total function $\sigma : \vars \to \exps(\vars)$, where $\exps(\vars)$ denotes symbolic expressions over~$\vars$. Let $\mathcal{U} \subseteq \vars$ denote the set of updated variables. For each $v \in \mathcal{U}$, $\sigma(v)$ denotes its final value after applying~$\mathcal{O}$; for each unchanged variable $v \notin \mathcal{U}$, we have $\sigma(v) = v$. The map is flat, meaning every expression $\sigma(v)$ is written directly in terms of the original variables in~$\vars$ without referencing intermediate assignments.
\label{def:eum}
\end{definition}

\textbf{Quadratic-Bounded Guarded Operators.} A \emph{Quadratic-Bounded Guarded Operator (QGO)} is any concrete operator $\op$ whose semantics can be represented by a pair $(\sigtwo, \guard)$, where $\sigma_{\leq 2}$ is a \emph{quadratic-bounded effective update map} and $\mathcal{G}$ is a \emph{conjunctive linear guard}. The EUM $\sigtwo : \vars \to \mathcal{P}_{\leq 2}(\vars)$ assigns each updated variable a polynomial of total degree at most two, i.e., $\sigtwo(v_j) = \sum_{i \leq k} a_{ik} v_i v_k + \sum_i b_i v_i + c$,  
with real coefficients $a_{ik}, b_i, c$, while the guard $\guard = \{ P v \le d \}$ is a conjunction of linear inequalities which restricts the input space.  
Either component may be absent, allowing only updates when $\guard$ is empty and only guards when $\sigtwo$ is the identity. This formulation covers all commonly occurring instructions handled by current libraries like \elina~\cite{SinghPV15}, including affine assignments, quadratic updates, guarded assignments, and conditional statements like \texttt{assume}. Moreover, a key benefit of this formulation is its expressiveness: it extends beyond single instructions to sequences whose combined effect can be flattened into a quadratic-bounded map, thereby enabling joint block-level analysis.
For example, the sequence $[x := a + b;\; y := x \cdot c]$ corresponds to $\sigtwo(x) = a + b$ and $\sigtwo(y) = (a + b) \cdot c$, respectively, which lies within the QGO class and can be analyzed jointly by the \algps algorithm.

\pheader{Template-Constrained Outputs}  
The \algps algorithm constructs a \textit{template-constrained} space of outputs, i.e. each output in the space is expressed over a fixed Template Constraint Matrix (TCM)~$\temp \in \tempspace$ (Sec.~\ref{sec:background}). For template-based domains such as Zones or Octagons, $\temp$ directly corresponds to the domain’s intrinsic constraint structure. In contrast, for expressive domains such as general Polyhedra, where templates are not fixed, any user-specified~$\temp$ can be provided; by default, the algorithm uses the Zone template. If the underlying abstract domain $\mathcal{D}$ supports abstract elements of arbitrary shapes, as in the Polyhedra domain, the generated outputs can be used directly. Otherwise, the domain must provide a \emph{sound conditional operator}, which soundly intersects abstract elements with arbitrary linear constraints. In such cases, the generated outputs can be used in two ways: (i) to refine the result of a baseline analyzer by intersecting its output with the generated output constraints, or (ii) independently by intersecting the constraints with~$\top$, yielding an abstract element $\absvalf{\mathcal{D}} \in \mathcal{D}$ that over-approximates $\absvalout$, i.e., $\gamma(\absvalout) \subseteq \gamma(\absvalf{\mathcal{D}})$.

\pheader{Parametric Space of Outputs} In template domains, an abstract element is of the form $\absval = \{\temp\mathbf{v} \ge \mathbf{c}\}$, characterized by scalar lower bounds~$\mathbf{c} = [c_1, \ldots, c_t]^T$, where each~$c_i$ corresponds to constraint~$\temp_i$. Instead of fixing these bounds to a single value, we allow them to vary as functions of parameters, thereby representing a space of possible abstract elements in a unified form. This is captured by \emph{Parametric Scalar Maps} (PSMs):

\begin{definition}[Parametric Scalar Map]
A Parametric Scalar Map (PSM) is a pair $\psm = (\parspace, \lf)$ where $\parspace \subseteq \real^d$ is a parameter space, and $\lf : \parspace \to \real$ maps each $\param \in \parspace$ to a finite scalar value.
\label{def:parameter-scalar-map}
\end{definition}

Each template constraint is assigned a PSM~$\psm_i$, which encodes a space of admissible values for the corresponding bound~$c_i$. A collection of such PSMs, one for each template direction, is used to define our parametric space of outputs as follows.




\begin{definition}[Parametric Space of Abstract Elements]
\label{def:parametric-space}
Let $\temp \in \tempspace$ be a fixed template matrix. Consider the ordered collection of $t$ PSMs $\overline{\psm} = [\psm_1, \ldots, \psm_t]$ where ~$\psm_i = (\parspace_i, \lf_i)$. Each $\psm_i$ is extended to a total map $\lf_i' : \real^{d_i} \to \real \cup \{-\infty\}$ defined as $\lf_i'(\param_i) = \lf_i(\param_i)$ if $\param_i \in \parspace_i$, and $\lf_i'(\param_i) = -\infty$ otherwise. The collection $\overline{\psm}$ induces the following parametric space of elements:
\[
\elemspace(\overline{\psm})
= \lambda \boldsymbol{\param}.
\;\{\temp_i \mathbf v \ge \lf_i'(\param_i)
\mid i = 1,\ldots,t \}
\]

Fixing choices of parameter $\param$, such as, ${\boldsymbol{\param}}=(\hat{\param}_1,\ldots,\hat{\param}_t)$ generates abstract elements
$\elemspace(\overline{\psm},{\boldsymbol{\param}})
= \{\temp_i \mathbf v \ge \lf_i'(\hat{\theta}_i)
\mid i=1,\ldots,t\}$.

\end{definition}

Intuitively, each PSM~$\psm_i$ specifies the finite values that the bound~$c_i$ can take.
When~$\param_i \notin \parspace_i$, $c_i$ defaults to~$-\infty$, effectively removing the constraint since it always holds. If~$\parspace_i$ is empty, $c_i$ has no finite value, no output can have finite $c_i$ and the entire space omits the corresponding constraint.

\begin{example} Suppose the template~$\temp$ includes two rows corresponding to the constraints $x + y \ge c_1$ and $x - y \ge c_2$. Consider the two Parametric Scalar Maps:
\[
\psm_1 = \left( \parspace_1 = \left\{ \lambda_1 \geq 1,\; \lambda_2 \geq 1,\; \lambda_1 + \lambda_2 \geq 1 \right\},\; \lf_1(\lambda_1, \lambda_2) = 2\lambda_1 + \lambda_2 \right),
\]
\[
\psm_2 = \left( \parspace_2 = \left\{ \lambda_3 \geq 1,\; \lambda_4 \geq 1,\; \lambda_3 + \lambda_4 \leq 1 \right\},\; \lf_2(\lambda_3, \lambda_4) = 3\lambda_3 + 2\lambda_4 \right)
\]
$\elemspace([\psm_1,\psm_2])$ defines a parametric space of elements with each parameter tuple~$(\lambda_1,\lambda_2,\lambda_3,\lambda_4)$ defining one element in this space. For $\psm_1$, feasible parameters yield finite bounds; for instance, $(\lambda_1,\lambda_2)=(1,1)$ gives $x+y\ge3$ and $(2,1)$ gives $x+y\ge5$.  
For $\psm_2$, no feasible parameters exist since $\lambda_3,\lambda_4 \ge 1$ and $\lambda_3+\lambda_4 \le 1$ cannot hold together, so every element in the space has $x-y \ge -\infty$.
\end{example}

A parametric space is \textit{sound space of outputs} for a given operator~$f$ and input element~$\absvalin$ if all possible instantiations (all abstract outputs in the space) over-approximate the effect of~$f$ on~$\absvalin$.

\begin{definition}[Soundness of a Parametric Space]
\label{def:soundness-space}
A parametric space~$\elemspace(\overline{\psm})$ for a template~$\temp$ is a sound space of outputs for an input abstract element~$\absvalin$ and concrete operator~$f$ if, for all parameter tuples~$\boldsymbol{\param}$,
$f(\gamma(\absvalin) \subseteq \gamma(\elemspace(\overline{\psm}, \boldsymbol{\param}))$.
\end{definition}

\begin{algorithm}[t]
\caption{Universal Parametric Output Space Encoder (\algps) Algorithm}
\label{algo:algps}
\begin{algorithmic}[1]
\State \textbf{Input:} Input~$\absvalin = \{A\mathbf{v} \le b\}$, QGO operator $\op = \{\sigtwo, \guard = P\mathbf{v} \le d\}$, Template $\temp \in \tempspace$
\State $\overline{\psm} \!\gets\! []$
\State $\efc \!\gets\! \{A\mathbf{v} \le b\} \cup \{P\mathbf{v} \le d\}$
\Comment{Effective Output Constraints (Input and Guard)}
\For{$i = 1$ to $t$}
  \State $\eff \!\gets\! \textsc{GetEffectiveObjective}(\temp_i, \sigtwo)$
  \Comment{Updated value of constraint $\temp_i$}
  \State $\psm_i \!\gets\! \textsc{ParametricLowerBound}(\eff, \efc)$
  \Comment{Duality-based parametric bounding of $\eff_i$}
  \State $\overline{\psm}.\texttt{append}(\psm_i)$
\EndFor
\State \Return $\elemspace(\overline{\psm})$
\Comment{Construct Parametric Space with the bounding PSMs}
\end{algorithmic}
\end{algorithm}

The \algps algorithm (Alg.~\ref{algo:algps}) constructs a sound parametric space of outputs in template~$\temp$ for an input abstract element~$\absvalin = \{A\mathbf{v} \le b\}$ from polyhedral domain~$\mathcal{D}$ and a QGO operator~$\op =\{\sigtwo, \; \guard = P\mathbf{v} \le d\}$. It works in two steps: first, it computes \textit{effective objectives and constraints} that capture how each template constraint transforms under~$\op$; second, it derives the PSM set $\overline{\psm} = [\psm_1, \ldots, \psm_t]$, where each PSM $\psm_i$ \textit{soundly lower bounds} the $i^{th}$ template constraint. The resulting set~$\overline{\psm}$ defines the parametric output space $\elemspace(\overline{\psm})$. We detail these steps below:

\pheader{Step 1: Computing Effective Objectives and Constraints}
First, \algps collects the \emph{effective constraints} $\efc$ by combining the constraints in the input $\{A\mathbf{v} \le b\}$ with those introduced by the operator $\{P\mathbf{v} \le d\}$, i.e.
$\efc = \{\mathbf{v} \mid \bar{A}\mathbf{v} \le \bar{b}\}$,  where  
$\bar{A} = \big[\! \begin{smallmatrix} A \\ P \end{smallmatrix} \!\big]$ and  
$\bar{b} = \big[\! \begin{smallmatrix} b \\ d \end{smallmatrix} \!\big]$.
Intuitively, $\efc$ captures the joint feasible region over which subsequent bounds are computed. Next, for each template constraint direction~$\temp_i$, \algps constructs an \emph{effective objective} $\eff_i$ that captures how the corresponding constraint evolves under the QGO~$O$. Each $\eff_i$ is obtained by applying the QGO’s update map~$\sigtwo$ to the linear form~$\temp_i\mathbf{v}$, yielding $\eff_i = \sum_j \temp_{i,j} \cdot \sigtwo(v_j)$. Since~$\temp_i$ is linear and~$\sigtwo$ produces expressions of degree at most two, each~$\eff_i$ remains quadratic-bounded. For example, let the input element be $\absvalin=\{\,0\le d\le2,\;0\le e\le1\,\}$, the operator~$O$ have updates $\sigtwo(x)=d^2$, $\sigtwo(y)=3d+e$, and guard $\mathcal{G}=\{\,d+e\le5\,\}$,  
and the template row be $\temp_i=(1,1)$ (i.e., $x+y$).  
Then the effective objective is $\eff_1=d^2+3d+e$, and the feasible region is 
$\efc=\{\,0\le d\le2,\;0\le e\le1,\;d+e\le5\,\}$.

\pheader{Step 2: Deriving Parametric Lower Bounds}
The most precise lower bound for template constraint $\temp_i$ can be computed by solving the minimization problem: $c_i^\# = \min_{\mathbf{v} \in \efc} \eff_i(\mathbf{v})$. But this is computationally expensive and returns only a fixed value. To obtain a space of sound bounds, we apply our \emph{Parametric Bounding Procedure} (App.~\ref{app:derivation}), which takes  $(\eff_i, \efc)$ as input and returns a \textit{sound lower-bound PSM} $\psm_i = (\parspace_i, \lf_i)$ satisfying  
$\lf_i(\param_i) \le c_i^\#$ for all $\param_i \in \parspace_i$. This means that each value produced by~$\lf_i$ forms a valid lower bound on~$c_i^\#$, compactly representing a space of sound bounds for the~$i^{\text{th}}$ template constraint. The procedure derives this space symbolically in three key steps:

\begin{enumerate}[leftmargin=*]

\item \textbf{Dual Construction with Domain-Aware Simplification:}
The first step is to construct the Lagrangian dual of the problem: $\min_{\mathbf{v} \in \efc} \eff_i(\mathbf{v})$. Since abstract domains already provide tight variable bounds, we treat these as box constraints and introduce Lagrange multipliers only for the remaining relational ones, yielding:  
$\max_{\boldsymbol{\lambda} \ge 0}\;\min_{\mathbf{v} \in [\mathbf{l}, \mathbf{u}]}\big(\eff_i(\mathbf{v}) + \boldsymbol{\lambda}^\top(\bar{A}\mathbf{v} - \bar{b})\big)$. This also works when some or all bounds are missing by using~$\pm\infty$ as limits. Here, the result of inner minimization is known as the dual function $g(\boldsymbol{\lambda})$. Considering the dual function has a key advantage: by the Weak Duality Theorem~\cite{boyd2004convex}, its value for any feasible~$\boldsymbol{\lambda}$ serves as a certified lower bound on the primal optimum $c_i^\#$, providing a principled way to obtain a space of sound lower bounds, parametrized by the dual variables $\boldsymbol{\lambda}$. Furthermore, the box-constrained structure allows $g(\boldsymbol{\lambda})$ to be computed symbolically and enables simplifications not possible in general quadratic programs (details of dual construction in Appendix ~\ref{subsec:eff-dual}).

\item \textbf{Symbolic Decomposition via Coefficient Splitting:}  
Another insight is that although solving the inner minimization defining~$g(\boldsymbol{\lambda})$ is difficult due to bilinear and indefinite quadratic terms, we can introduce symbolic parameters~$(\mathbf{S}, \mathbf{D})$ to \emph{redistribute linear coefficients} across the objective:
\[
\underbrace{\min_{\mathbf{v}}
\big(
\mathbf{v}^\top H \mathbf{v} + \mathbf{q}^\top \mathbf{v}
+ \boldsymbol{\lambda}^\top A \mathbf{v}
\big)}_{g(\boldsymbol{\lambda})}
\;\rightsquigarrow\;
\underbrace{\min_{\mathbf{v}} \sum_{i<j} (H_{ij} v_i v_j + D_{ij}^i v_i + D_{ij}^j v_j)
+ \min_{\mathbf{v}} \sum_i (Q_i v_i^2 + S_i v_i)
+  \cdots}_{\hat{g}(\boldsymbol{\lambda}, \mathbf{S}, \mathbf{D})}
\]
Here, $H_{ij}$ and $Q_i$ are fixed coefficients from the original quadratic form, while $D_{ij}^i$, $D_{ij}^j$, and $S_i$ are symbolic split parameters that control the redistribution of linear terms. This coefficient splitting decomposes the objective into independent 1D and 2D subproblems that can be minimized efficiently in isolation (Appendix ~\ref{subsec:inner-min}). The decomposition is sound: by minimizing each subproblem independently, we obtain a value less than or equal to the true joint minimum, thus yielding a valid lower bound on the original objective: $\forall\, \boldsymbol{\lambda} \!\ge\! 0,\, \mathbf{S}, \mathbf{D}.\;
\hat{g}(\boldsymbol{\lambda}, \mathbf{S}, \mathbf{D})
\le g(\boldsymbol{\lambda}) \le c_i^\#.$

\item \textbf{Tractable, Modular Evaluation of Subproblems:}  
Each resulting subproblem is now either a linear, bilinear, or quadratic expression over a box (1D or 2D), which can be evaluated using closed-form expressions or cheap bounded minimization routines. Some subproblems may have unbounded minima over all~$\boldsymbol{\lambda} \ge 0$, but we aim to capture the space of non-trivial (finite) bounds and thus derive additional conditions on~$\boldsymbol{\lambda}$ under which they minimum remains bounded, and add these to the feasible region~$\parspace_i$. The resulting symbolic expression for minimum is captured in $\lf_i$, thus yielding the lower bounding PSM $\psm_i = \{\parspace_i, \lf_i\}$ (Appendix~\ref{subsec:inner-parts-eval}).
\end{enumerate}

Having defined the construction of the PSMs~$\overline{\psm}$ that collectively define the parametric output space $\elemspace(\overline{\psm})$, we now formally establish the soundness and precision guarantees.

\begin{theorem}[Soundness of the Constructed Parametric Space]
\label{thm:soundness}
Let $\elemspace(\overline{\psm})$ be the parametric space of elements constructed by the \algps algorithm for an input element $\absvalin$, QGO operator $\op$ and template $\temp$, where $\overline{\psm} = [\psm_1, \ldots, \psm_t]$ is the collection of PSMs $\psm_i$ obtained by lower bounding effective objectives $\eff_i$.
Then, $\elemspace(\overline{\psm})$ is a sound space of outputs (Def. \ref{def:soundness-space}) for~$\mathtt{a}_{\mathit{in}}$ and~$\op$.
\end{theorem}

\begin{proof}[Proof]
Assume that $\elemspace(\overline{\psm})$ is not a sound space of outputs for $\absvalin$ and $\op$. Then, by Definition~\ref{def:soundness-space}, there exists $\hat{\boldsymbol{\param}}$ such that $f(\gamma(\absvalin)) \nsubseteq \gamma(\elemspace(\overline{\psm},\hat{\boldsymbol{\param}}))$. For the QGO $\op = \{\sigtwo, \guard\}$, this means there is $\mathbf{v} \in \gamma(\absvalin)$ with $\guard(\mathbf{v})$ and
$\mathbf{v}' = \sigtwo(\mathbf{v})$ such that $\mathbf{v}' \notin \gamma(\elemspace(\overline{\psm},\hat{\boldsymbol{\param}}))$.
This implies that some template constraint~$j$ is violated, i.e., $\temp_j\mathbf{v}' > \lf_j(\hat{\param}_j)$. By Theorem~\hyperref[thm:soundness-fspb]{C.1} in Appendix, we have $\lf_j(\hat{\param}_j) \le c_j^\#$ (each PSM encodes lower bounds on the minimum value of its  template constraint), and since $c_j^\#$ is the minimum of $\temp_j\mathbf{v}'$ over all reachable states, it follows that $c_j^\# \le \temp_j\mathbf{v}'$, leading to a contradiction.
Thus, such a violation cannot occur, and $\elemspace(\overline{\psm})$ is sound space of outputs for~$\absvalin$ and~$\op$.
\end{proof}

\begin{theorem}[Most Precise Output for Linear Case]
\label{thm:lin-precise}
For QGO operators $\op$ with linear $\sigtwo$, the space
$\elemspace(\overline{\psm})$ contains the most precise abstract output.
That is, for each PSM $\psm_i$, there exists a parameter instantiation
that yields the most precise bound $c_i^\#$.
\end{theorem}

\begin{proof}[Proof]
This follows by Theorem~\hyperref[thm:lin-tightness]{C.2} in Appendix, which uses Strong Duality~\cite{boyd2004convex} to show that, for linear cases, the induced space captures the optimum, i.e., there exists $\theta \in \Theta$ such that $L(\theta) = c_i^\#$.
\end{proof}

\textbf{Other key properties of \algps.}
Beyond soundness and precision, the parametric space $\elemspace(\overline{\psm})$ constructed by our \algps algorithm has the following key features:
\begin{enumerate}[leftmargin=*]
    \item \textbf{Interval Relaxation at $\vec{0}$}: The abstract element $\elemspace(\overline{\psm}, \vec{0})$, obtained by setting $\boldsymbol{\param} = \vec{0}$, corresponds to the output of the \emph{interval relaxation} heuristic used commonly in existing libraries, which computes bounds by simply discarding inter-variable constraints and only using the interval bounds of the variables (By Theorem \hyperref[thm:interval-start]{C.2} proved in Appendix).
    \item \textbf{Differentiable, Polyhedral Parameterization:}
    All PSMs $\psm_i = (\parspace_i, \lf_i)$ are such that the parameter space $\parspace_i$ is polyhedral, i.e., of the form $\{\mathbf{A}_i\param_i \le \mathbf{b}_i\}$ 
    (Theorem~\hyperref[thm:theta-polyhedral]{C.4} in Appendix), and each map $\lf_i(\theta_i)$ is piecewise differentiable 
    (Theorem~\hyperref[thm:L-differentiable]{C.5} in Appendix). 
    This structure makes the space easy to traverse and enables the use of gradient-based optimization techniques.
    \item \textbf{Extensibility to Joins and Disjunctions:} The construction extends naturally to joins and disjunctions by constructing per-branch parametric spaces and combining them via a $\min$ over the joint parameter space, which remains polyhedral and piecewise differentiable, allowing AGG to be applied directly (details in Appendix \hyperref[app:joins]{E}).
\end{enumerate}

\section{Gradient-Guided Exploration of Parametric Output Space}
\label{sec:agg}

Once we have a space of sound outputs $\elemspace(\overline{\psm})$, the next step is to traverse this space and find a parameter vector $\boldsymbol{\param}$, which yields the abstract element $\elemspace(\overline{\psm}, \boldsymbol{\param})$ suitable for analysis. Each PSM $\psm_i = \{\parspace_i, \lf_i\}$ defines the parameter space $\parspace_i$ to search, where each $\param_i \in \parspace_i$ yields a finite bound for the $i^{th}$ template direction. Template directions with empty $\parspace_i$ always have the trivial $-\infty$ bound and don't require search. Any efficient search procedure should be guided by the analysis requirements of the downstream task, for which it must account for the following components:

\begin{enumerate}[leftmargin=*]

\item \textbf{Search Objective ($\mathcal{J}$):}
The search procedure should allow downstream tasks to specify a score function $\mathcal{J}$ that measures the \textit{quality} of candidate outputs, where higher scores correspond to outputs that better suit the analysis goal. For each constraint~$i$, the score function assigns a quality score~$\mathcal{J}(L_i(\theta_i))$ to every possible lower bound in the set~$\mathcal{L}_i = \{L_i(\theta_i) \mid \theta_i \in \Theta_i\}$. This score can help guide the search process by steering it toward parameter values~$\theta_i$ that yield more desirable outputs, i.e., pick lower bounds $L_i(\theta_i)$ with higher $\mathcal{J}(L_i(\theta_i))$. Some examples of such score functions include:
\begin{enumerate}
    \item \emph{Precision Objective $\mathcal{J}_{\mathit{prec}}$:}
    In general program analysis settings, the goal is to infer the most precise invariants. For example, consider a constraint of the form $x - y \geq L(\theta)$ for $\theta \in \Theta$. The most precise invariant is obtained by choosing $\theta^* \in \Theta$ that maximizes $L(\theta)$, since $x - y \geq c_1$ is strictly more precise than $x - y \geq c_2$ whenever $c_1 > c_2$. This reflects the general principle that larger lower bounds yield tighter (i.e., more precise) abstract outputs. Accordingly, one can define the score function as $\mathcal{J}_{\mathit{prec}}(\theta_i) = L_i(\theta_i)$, so that maximizing $\mathcal{J}_{\mathit{prec}}$ directly corresponds to improving precision.

    \item \emph{Target Objective $\mathcal{J}_{\mathit{tgt}}$:}
    Beyond computing the most precise invariants, program analysis tasks may require selecting outputs that satisfy a given specification. For example, consider a constraint of the form $x + y \geq L(\theta)$ for $\theta \in \Theta$, and suppose the goal is to prove $x + y \geq e$. This reduces to finding $\theta^* \in \Theta$ such that $L(\theta^*) \ge e$. Accordingly, one can define a score function that measures satisfaction of this condition, such as the violation-based score $\mathcal{J}_{\mathit{tgt}}(\theta_i) = -\max(0,\, e_i - L_i(\theta_i))$, so that higher scores correspond to smaller violations and the maximum is attained when $L_i(\theta_i) \ge e_i$. This example also highlights how the search objective $\mathcal{J}$ enables adaptive search: rather than always computing the most precise output, it can guide the search toward outputs \textit{sufficient} to establish the desired property or meet the task requirements. Techniques proposed by works like DL2~\cite{dl2} can further be used to encode more complex specifications as differentiable objectives.
\end{enumerate}

\item \textbf{Runtime Control ($\mathcal{R}$):} The search procedure should also expose mechanisms that allow tasks to regulate the computational effort during search by defining some form of runtime budget $\mathcal{R}$. This enables tasks with different runtime constraints to choose appropriate configurations. Examples of such budgets can include search parameters, such as the number of samples, gradient steps, or epochs, or more coarse-grained options like wall-clock timeouts. Varying this budget can allow the tasks to find the desired tradeoff between output quality and runtime.
\end{enumerate}

\textbf{Limitations of Existing Search Strategies.}
A naive search strategy is to generate outputs by randomly sampling parameters from the feasible space $\parspace$. Runtime $\mathcal{R}$ can be controlled by limiting the number of samples, and the best output is selected by measuring the score function $\mathcal{J}$. However, this search is unguided and not  \textit{objective-aware}, often failing to find high-scoring outputs. Moreover, efficient sampling is difficult when $\parspace$ is high-dimensional or defined by complex constraints. On the other hand, a highly objective-aware strategy is to use solvers to directly maximize $\mathcal{J}_i(\lf_i(\param_i))$ over each $\parspace_i$.  This generates high-quality outputs but is expensive in practice, as it requires solving one optimization problem for each constraint direction $i$, resulting in many solver invocations. Although solvers provide timeout mechanisms, these apply per call, while downstream tasks usually impose a global runtime budget $\mathcal{R}$. When $\mathcal{R}$ is split across many calls, as in big programs with many template constraints, each call gets little time, often yielding no feasible or low-quality outputs. In summary, random sampling provides controllable runtime but lacks objective guidance, whereas solver-based optimization offers strong objective guidance but limited runtime control.

\pheader{Gradient-Guided Exploration}
To balance objective awareness with runtime control, an ideal search procedure should progressively move toward parameters that yield higher scores under the objective $\mathcal{J}_i$. 
Gradient-guided search naturally fits this goal: it enables incremental improvement, avoids costly solver calls or exhaustive enumeration, and improves output quality as more runtime becomes available. This is especially well-suited for our task, because as discussed in the last section,
the output $\lf_i(\param_i)$ is piecewise differentiable. Assuming the score function $\mathcal{J}_i$ can be defined differentiably, as is the case for many tasks including precision objective, we can follow the gradient $\nabla \mathcal{J}_i(\lf_i(\param_i))$ to optimize it. 
This results in an iterative procedure where the number of gradient steps serves as a natural and fine-grained control for the runtime $\mathcal{R}$. Unlike solver-based methods that require separate invocations per constraint, gradient descent 
enables joint optimization across all template directions, progressively improving output quality within the available budget.

\begin{algorithm}[t]
\caption{Adaptive Gradient Guidance (AGG) Search}
\label{algo:agg}
\begin{algorithmic}[1]
\State \textbf{Input:} score $\mathcal{J}(L(\theta))$, feasible region $\Theta = \{\hat{\mathbf{A}}\theta \le \hat{\mathbf{b}}\}$, step $\eta$, penalty $\beta$, gradient steps $R$
\State \textbf{Init:} $\theta \gets 0$, $\theta^{\text{best}} \gets 0$ \Comment{Start from interval relaxation}
\For{$r = 1$ to $R$}
  \If{$\theta \in \Theta$}
    \State $\theta \gets \theta + \eta \nabla \mathcal{J}(L(\theta))$
    \Comment{Ascend along objective inside feasible region}
  \Else
    \State $\theta \gets \theta - \eta \beta \nabla \|\max(\hat{\mathbf{A}}\theta - \hat{\mathbf{b}}, 0)\|_p$
    \Comment{Reduce violations outside feasible region}
  \EndIf
  \If{$\theta \in \Theta$ and $\mathcal{J}(L(\theta)) > \mathcal{J}(L(\theta^{\text{best}}))$}
    \State $\theta^{\text{best}} \gets \theta$
    \Comment{Record best feasible parameters found so far}
  \EndIf
\EndFor
\State \Return $\theta^{\text{best}}$
\end{algorithmic}
\end{algorithm}

However, applying standard descent methods in our setting presents several challenges. 
We have a constrained maximization problem: maximize score $\mathcal{J}(\lf(\param))$ over the parameter space $\parspace$, which is polyhedral (as discussed in the previous section) and so defined as $\hat{A}\param \le \hat{b}$. A naive gradient step $\param \leftarrow \param + \eta \nabla \mathcal{J}(\lf(\param))$ may leave the feasible region $\parspace$. 
Projected Gradient Descent (PGD) corrects this by projecting back onto $\parspace$, 
but for polyhedral regions, each projection requires solving a quadratic or linear program, making repeated updates computationally expensive. Penalty-based methods avoid projection by adding a penalty term such as $-\beta \lVert \max(\hat{A}\param - \hat{b}, 0) \rVert_p$ to the objective. 
However, this assumes $\mathcal{J}(\lf(\param))$ is defined for all $\param \in \mathbb{R}^d$, 
whereas $\lf(\param)$ is only valid within $\parspace$, making the objective ill-posed outside it. Barrier methods, such as log-barrier or interior-point techniques, maintain feasibility by penalizing boundary violations. But they require strict feasible initialization and solve large linear problems at each iteration, making them unscalable. To overcome these challenges, we propose a lightweight, projection-free algorithm called \emph{Adaptive Gradient Guidance (AGG)} (Alg.~\ref{algo:agg}). AGG avoids solver calls and does not require a strictly feasible start. It adaptively selects the update direction based on whether the current point $\theta$ lies inside the feasible region $\Theta$. 
When feasible, it ascends along the gradient of the score to improve output quality; 
when infeasible, it follows a penalty gradient that reduces constraint violations to restore feasibility.

\pheader{Monotonic Progress Over Baseline}
The AGG algorithm enables objective-aware progress within $\Theta$ while steering infeasible iterates back without projections or solver calls. As discussed in the previous section, $\theta = 0$ corresponds to the interval-relaxation output used by existing libraries, so we start the search from this point as a natural baseline. The algorithm tracks the best feasible parameters after each gradient step, making the resulting progress \emph{monotonic}: the returned output is never weaker than the baseline, and increasing the search budget can only preserve or improve the result. After the allotted number of gradient steps, the procedure returns the best-scoring parameters found. For linear cases, since the parametric space contains the most precise output (Theorem~\ref{thm:lin-precise}), AGG can, in principle, recover this output as well. This is demonstrated in Section~\ref{sec:eval1}, where our gradient-guided search efficiently finds the most precise outputs in practice.
\section{\toolname Framework}
\label{sec:framework}

We implement the ideas discussed above in \toolname\footnote{https://github.com/uiuc-focal-lab/AbsEvolve} framework. It is built upon the \clam\footnote{https://github.com/seahorn/clam} analyzer (part of SeaHorn~\cite{seahorn} framework) and the numerical abstract domain library \elina~\cite{SinghPV17}. As discussed above, the \algps algorithm works for a wide range of QGO operators. For our implementation, we focus on on assignment operators, which are common in real-world programs. This includes assignment instructions that fall in QGO class such as (1) Affine assignments: $v_j := \sum_i c_i v_i + b$, with $c_i, b \in \mathbb{R}$ and (2) Quadratic assignments: $v_j := \sum_{i \leq k} a_{ik} v_i v_k + \sum_i b_i v_i + c$, with $a_{ik}, b_i, c \in \mathbb{R}$. Moreover, there can also be assignment sequences whose combined effect lies within the QGO class. The \toolname framework therefore (1) merges such sequences into \textit{blocks} for joint and more precise analysis, and (2) performs abstract interpretation on the transformed program.

\pheader{Merging Assignment Sequences into Blocks}
Frameworks like \clam analyze programs instruction-by-instruction, as implementing abstract transformers for entire sequences is challenging. In contrast, \toolname uses a merging algorithm that first identifies assignment sequences~$\seq$ whose combined updates remain within the QGO class; that is, sequences whose Effective Update Map (Definition~\ref{def:eum}) is quadratic-bounded. Each such sequence is then merged into a single block~$\block(\seq)$ for joint analysis. For example, as illustrated in Figure~\ref{fig:split-eam}, \toolname identifies two quadratic-bounded sub-sequences~${\seq}_1$ and~${\seq}_2$ within the larger sequence~$S$, and merges them into blocks. Note that if all the assignments were merged together, the resulting EUM would have $\sigma(z) = (2 * a * b * c - a)$ which is cubic and violates the quadratic-bounded restriction. Hence, the merging algorithm splits the sequence.
The algorithm also provides options such as only merging quadratic assignments or setting a maximum block length (detailed algorithm in Appendix~\hyperref[algo:merge-blocks]{B}).

\begin{figure}[t]
\centering
\renewcommand{\arraystretch}{1}
{\normalsize
\[
\begin{array}{@{}r@{\quad}l@{\quad\quad}c@{\quad}c@{}}
S: &
\begin{array}{@{}l@{}}
\texttt{x := 2 * a} \\[-0.3ex]
\texttt{y := x * b} \\[-0.3ex]
\texttt{z := y * c - a} \\[-0.3ex]
\texttt{w := z + d}
\end{array}
&
\begin{aligned}
  \seq_1 \; &=\;
      \begin{array}{@{\hskip 1pt}l@{\hskip 1pt}}
        \texttt{x := 2 * a} \\[-0.3ex]
        \texttt{y := x * b}
      \end{array}
  \quad &\sigma_{\block(\seq_1)} \; =\;
    \left\{
      \begin{array}{@{}l@{\;\mapsto\;}l@{}}
        \texttt{x} \\[-0.3ex]
        \texttt{y}
      \end{array}
      \quad
      \begin{array}{@{}l@{}}
        \texttt{2 * a} \\[-0.3ex]
        \texttt{2 * a * b}
      \end{array}
    \right\}
  \\[1ex]
  \seq_2 \; &=\;
      \begin{array}{@{\hskip 1pt}l@{\hskip 1pt}}
        \texttt{z := y * c - a} \\[-0.3ex]
        \texttt{w := z + d}
      \end{array}
  \quad &\sigma_{\block(\seq_2)} \; =\;
    \left\{
      \begin{array}{@{}l@{\;\mapsto\;}l@{}}
        \texttt{z} \\[-0.3ex]
        \texttt{w}
      \end{array}
      \quad
      \begin{array}{@{}l@{}}
        \texttt{y * c - a} \\[-0.3ex]
        \texttt{y * c - a + d}
      \end{array}
    \right\}
\end{aligned}
\end{array}
\]
}

\caption{While merging assignment sequences, \toolname partitions the sequence $S$ into two sub-sequences and merges them into blocks, each of whose Effective Update Map remains quadratic-bounded. Without this partitioning, the combined update would exceed quadratic degree, violating the QGO restriction.}
\label{fig:split-eam}
\end{figure}

\pheader{Abstract Interpretation on the Transformed Program}
The merging of assignment sequences $\seq$ into blocks $\block(\seq)$ in a program $\mathcal{P}$ yields a transformed program $\mathcal{P}'$. \toolname performs abstract interpretation on this new program $\mathcal{P}'$. For affine assignments, quadratic assignments and blocks $\block(\seq)$, it uses our evolving transformer by first applying the \algps algorithm to construct a space of sound outputs and then using the \alggs algorithm to search this space via gradient-descent and select a output. This search can be controlled by user-configurable parameters such as step size~$\eta$ and number of gradient steps~$R$ used by the \alggs algorithm. Standard transformers from \elina are used for the rest of the instructions.

\toolname analyzes the transformed program $\mathcal{P}'$ but the goal is to obtain invariants for the original program $\mathcal{P}$. We argue this is sound in two steps. Let $\llbracket \mathcal{P} \rrbracket_\ell$ denote the set of concrete states reachable at program point (label) $\ell$ in $\mathcal{P}$. First, for every concrete operator in $\mathcal{P}'$, the analysis uses a sound abstract transformer: our evolving transformer for assignments and blocks (sound by Thm.~\ref{thm:soundness}) and standard \elina transformers otherwise. So, at each program point $\ell$, if the analysis computes abstract output $\absvalf{l}$, then $\llbracket \mathcal{P}' \rrbracket_\ell \subseteq \gamma(\absvalf{l})$ holds. Next, by Theorem~\ref{thm:sound-eum}, we have for each program point $\ell$ that $\llbracket \mathcal{P} \rrbracket_\ell \subseteq \llbracket \mathcal{P}' \rrbracket_\ell$. This, combined with the first step (per-step abstract soundness on $\mathcal{P}'$, i.e., $\llbracket \mathcal{P}' \rrbracket_\ell \subseteq \gamma(\absvalf{\ell})$), yields $\llbracket \mathcal{P} \rrbracket_\ell \subseteq \gamma(\absvalf{\ell})$, thus proving the soundness of our analysis.

\begin{theorem}[Soundness of Block Merging (per point)]
\label{thm:sound-eum}
Let $\mathcal{P}'$ be obtained from $\mathcal{P}$ by replacing each assignment sequence $\seq$ with a block $\block(\seq)$ whose concrete semantics is its Effective Update Map $\sigma_{\block(\seq)}$. Then, for every program point (label) $\ell$, $\llbracket \mathcal{P} \rrbracket_\ell \;\subseteq\; \llbracket \mathcal{P}' \rrbracket_\ell$.
\end{theorem}

\begin{proof}[Proof Sketch]
Using Theorem~\hyperref[thm:eum-over]{C.6} in Appendix, which proves that jointly analyzing a sequence via its EUM over-approximates step-wise semantics, we have $\llbracket \seq \rrbracket \subseteq \llbracket \block(\seq) \rrbracket$ for all replaced sequences $\seq$. As composition of semantics is monotonic, replacing an instruction by its over-approximation can only add new reachable states, and so $\ell$, $\llbracket \mathcal{P} \rrbracket_\ell \;\subseteq\; \llbracket \mathcal{P}' \rrbracket_\ell$.
\end{proof}
\section{Evaluation}
\label{sec:eval}

\pheader{Research Questions} Our evaluation addresses two research questions that assess the key benefits of our evolving transformer: its ability to flexibly trade off precision and efficiency, achieve high precision efficiently, and generalize across domains and instructions. 
\begin{itemize}[leftmargin = *]
\item \textbf{RQ1:} Can \toolname provide controllable precision–efficiency tradeoffs while also achieving the most precise invariants for linear operators efficiently?
\item \textbf{RQ2:} Does \toolname maintain this adaptability and precision for more complex cases such as quadratic operators and the Polyhedra domain?
\end{itemize}

\textbf{Experimental Setup \& Benchmarks.}
We use our \toolname framework (Sec.~\ref{sec:framework}) for experiments and compare the invariants computed by it with those computed by the state-of-the-art library \elina~\cite{SinghPV17}. All experiments were run on a machine with a 16-core Intel(R) Core(TM) Ultra 7 255H CPU @ 1.10,GHz, 31 GiB RAM, and Ubuntu 24.04.3 LTS.
As benchmarks, we use the 57 programs from the \textsc{NLA-Digbench}~\cite{digbench} suite of the SV-COMP benchmarks~\cite{svcomp}. This is one of the largest suites of programs involving nonlinear numerical invariants and is widely used~\cite{nla1, nla2, nla3} to evaluate approaches that compute numerical invariants. These benchmarks include a wide range of instructions, such as affine and quadratic assignments and their sequences, which makes them challenging for standard analyses and also allows us to test how well \toolname handles complex linear and nonlinear operators. We analyze these programs across three popular numerical domains, Zones, Octagons, and Polyhedra, and keep a timeout of 200s for each program.

\subsection{Efficient and Adaptable Analysis for Linear Operators}
\label{sec:eval1}

As discussed in Section~\ref{sec:framework}, our evolving transformer can handle both affine and quadratic operators and their sequences. However, obtaining ground-truth most-precise outputs for quadratic operators is infeasible and so for this section, we restrict \toolname to only use our evolving transformer for affine assignments and their sequences. This allows direct comparison with the most-precise outputs for affine operators. Also, since the Polyhedra domain already captures affine assignments exactly, we focus only on the Zones and Octagon domains in this section. For each domain, we first analyze all benchmark programs using \elina to compute the set of baseline invariants $\invs_{bl}$, collected at the entry of each basic block across all programs (237 in total). This analysis takes around 20s for both domains. Next, we compute the \textit{ground-truth} invariants $\invs_{gt}$ obtained using the most-precise transformer for affine assignment sequences. For linear operators, this involves solving $t$ linear programs (LPs) per transformer call (Sec.~\ref{sec:upose}), where $t$ is the number of template directions, which grows quadratically with the number of variables for Zones and Octagon domains. We use the LP-solving-based transformer~\cite{tcm_domain} by directly using Gurobi~\cite{gurobi} solver to solve these LPs efficiently. As discussed in Sec.~\ref{sec:agg}, \toolname jointly searches for bounds over all template directions, while the LP solver optimizes each direction independently. For fairness, we parallelize these LP calls across available threads. This method computes the set of most-precise \textit{ground-truth} invariants $\invs_{gt}$, taking 230s for Zones and 350s for Octagons. We also experimented with Symba~\cite{awsSolver}, an SMT-based optimizer for Linear Real Arithmetic problems, but its reliance on SMT solving made it significantly slower, taking 43 minutes for the Zones domain (timing out on 11 programs) and 56 minutes for the Octagon domain (timing out on 12).

\begin{figure}
    \centering

    \begin{subfigure}[b]{0.47\linewidth}
        \centering
        \includegraphics[width=\linewidth]{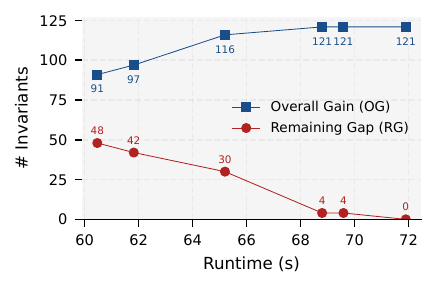}
        \includegraphics[width=0.87\linewidth]{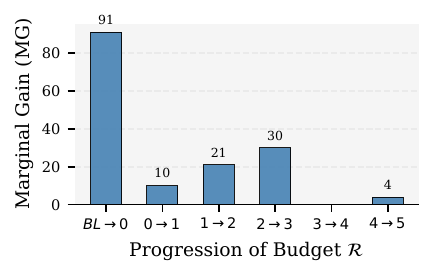}
        \caption{Zones}
        \label{fig:zones-linear}
    \end{subfigure}
    \hfill
    \begin{subfigure}[b]{0.47\linewidth}
        \centering
        \includegraphics[width=\linewidth]{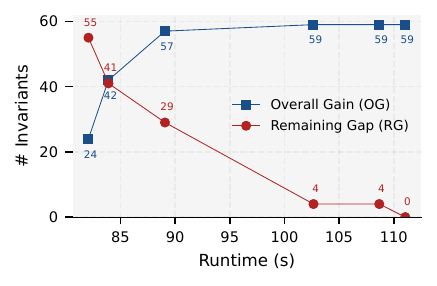}
        \includegraphics[width=0.87\linewidth]{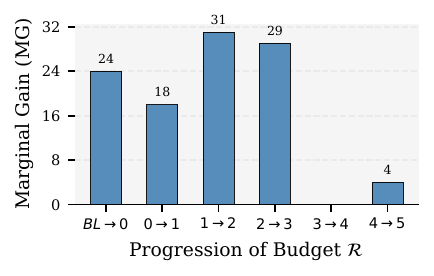}
        \caption{Octagons}
        \label{fig:octagons-linear}
    \end{subfigure}

\caption{Evolution of invariant strengthening for Zones and Octagons (linear case) as the gradient-step budget $\mathcal{R}$ varies, showing how the analysis can be adapted across settings. Analysis with smaller $\mathcal{R}$ runs faster but is less precise, while larger $\mathcal{R}$ yields stronger invariants, with RG reaching $0$ (most-precise invariants) at $\mathcal{R}=5$.}
    \label{fig:zones-oct-linear}
\end{figure}

\pheader{Monotonic Strengthening of Invariants}
Next, we analyze the benchmark programs using \toolname by varying the search
budget $\mathcal{R}$ from~0 to~5 (where $\mathcal{R}$ is the number of gradient
steps per transformer call) with $\mathcal{J}$ set to the precision objective $\mathcal{J}_{prec}$ (Section~\ref{sec:agg}). Let $\invs_r = \{a_r^1, a_r^2, \dots\}$ denote the invariants produced at budget $r$, where $a_r^i$ is the $i$-th corresponding
invariant in the set. These sets exhibit a monotonic strengthening structure:
for all $r \ge 1$ and all $i$, we have $\gamma(a_r^i) \subseteq
\gamma(a_{r-1}^i)$, so invariants never become weaker as the budget increases;
this property is enforced by the monotonicity of AGG discussed in
Section~\ref{sec:agg}. Since $r=0$ performs no gradient steps, AGG returns the interval–relaxation
baseline output used by standard libraries such as \elina, so one might expect
$\invs_0 = \invs_{bl}$. However, because \toolname performs
sequence-level reasoning (rather than instruction-by-instruction as in the
baseline), $\invs_0$ often already strengthens several baseline invariants,
i.e., $\gamma(a_0^i) \subseteq \gamma(a_{bl}^i)$. More generally, for all $r$
and $i$, we have $\gamma(a_r^i) \subseteq \gamma(a_0^i) \subseteq
\gamma(a_{bl}^i)$. To more precisely capture this strengthening behavior across budgets, we track the following metrics:

\begin{enumerate}[leftmargin=*]
    \item \textbf{Overall Gain (OG)}: This counts how many invariants in $\invs_r$ are stronger than their counterparts in $\invs_{bl}$, i.e. $OG = \sum_{i=1}^{|\invs_r|}\mathbf{I}(\gamma(\absvalf{r}^i) \subset \gamma(\absvalf{bl}^i))$, where $\mathbf{I}(\cdot)$ denotes the indicator function.

    \item \textbf{Remaining Gap (RG)}: This measures how many invariants in $\invs_r$ are weaker than those in ground-truth set $\invs_{gt}$ and gives a sense of gap from optimality: $RG = \sum_{i=1}^{|\invs_r|}\mathbf{I}(\gamma(\absvalf{gt}^i) \subset \gamma(\absvalf{r}^i))$.

    \item \textbf{Marginal Gain (MG)}: OG provides a global view of total improvement over $\invs_{bl}$ but can hide progress across budget settings, because once an invariant strengthens over the baseline, it contributes to OG only once, even if it is refined further at higher budgets. MG provides a local view by measuring how many invariants are strengthened at budget $r$ relative to $r{-}1$, i.e., $MG = \sum_{i=1}^{|\invs_r|} \mathbf{I}(\gamma(a_r^i) \subset \gamma(a_{r-1}^i))$. $\invs_{r-1}$ is taken as $\invs_{bl}$ when $r = 0$. MG shows how much additional refinement a one-step increase in search budget provides over the previous setting.
\end{enumerate}

\pheader{Adaptable Precision–Efficiency Analysis}
Figure~\ref{fig:zones-oct-linear} shows the  metrics defined above over different gradient-step budgets. The line charts show the runtimes for budgets $\mathcal{R} = 0$ to $5$ (dots left to right) along with how the overall gains from baseline (OG) and the gap to the most-precise invariants (RG) evolve as the budget increases. The accompanying bar plots report the marginal gains obtained when increasing the budget from one setting to the next. As expected, larger budgets lead to higher runtimes since more gradient steps are performed within each transformer call. However, they also yield increasingly precise invariants, as reflected by rising OG and decreasing RG trends. Although OG flattens after $\mathcal{R}=3$, the continued decrease in RG and non-zero marginal gains from $\mathcal{R}=4$ to $\mathcal{R}=5$ show that further strengthening still occurs on invariants already strengthened. Overall, these results show that the number of gradient-steps $R$ provides an effective way to adapt the analysis by tuning the precision–efficiency tradeoff: larger budgets improve precision at higher cost, while smaller budgets offer faster but less precise analyses.

\pheader{Computing the Most-Precise Invariants Efficiently}
As demonstrated by Figure~\ref{fig:zones-oct-linear}, at $\mathcal{R}=5$ \toolname achieves $RG=0$ for both Zones and Octagons, meaning that the invariant set $\invs_5$ matches the ground-truth most-precise set $\invs_{gt}$. 

\begin{wrapfigure}{r}{0.40\linewidth}
    \centering
    \vspace{-1.5em}
    \includegraphics[width=\linewidth]{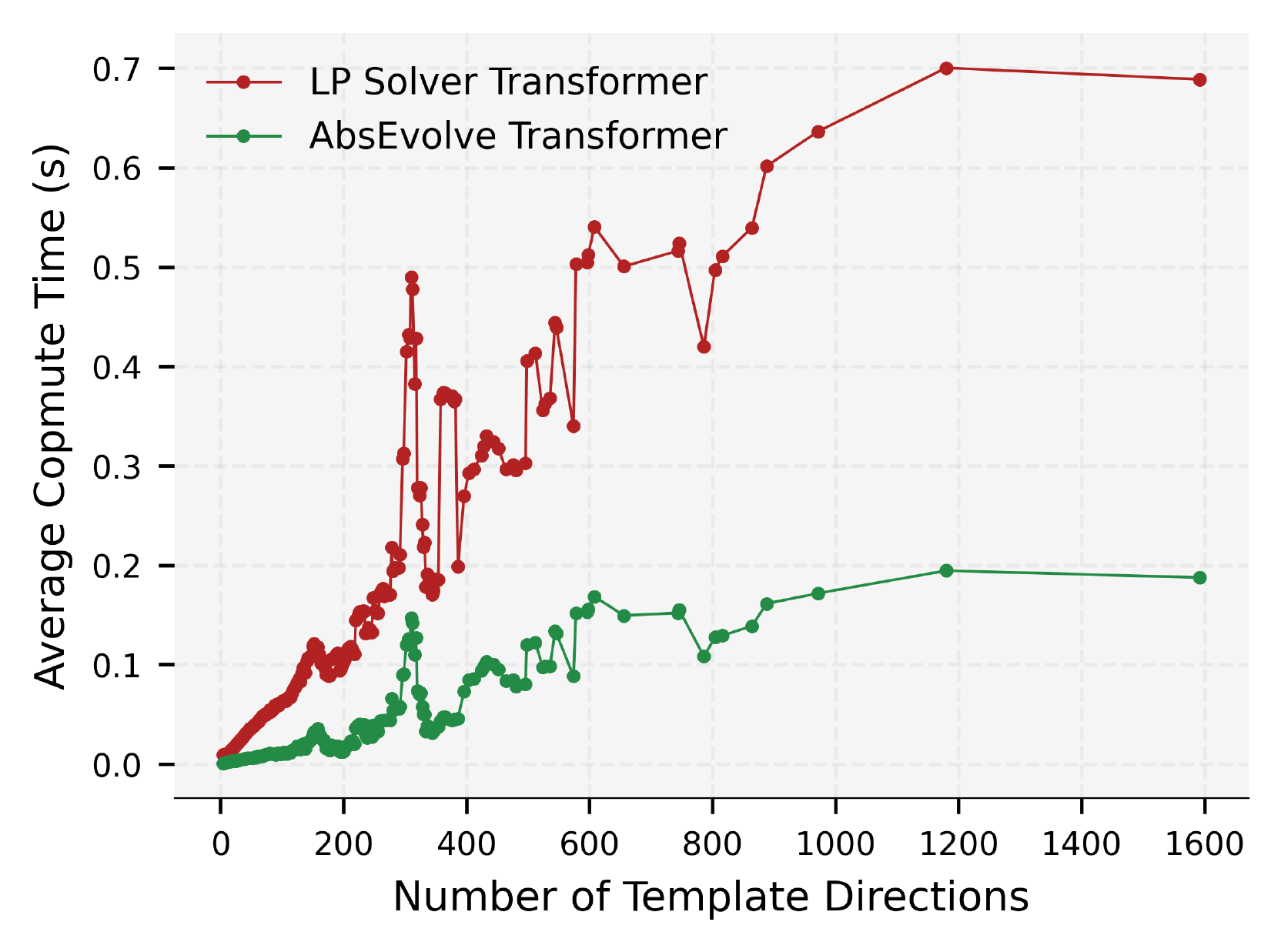}
    \vspace{-2em}
    \caption{Transformer Computer Time vs. Number of Template Directions}
    \label{fig:solver-time}
\end{wrapfigure}

\noindent Notably, \toolname reaches the most-precise invariants in only \textit{72s} for Zones and \textit{111s} for Octagons, compared to \textit{230s} and \textit{350s} for the LP-solver-based transformer, making it about \textbf{3.2$\times$ faster}. As discussed in Sec.~\ref{sec:upose}, the parametric output space contains the most-precise output for linear cases, and these results show that the \alggs algorithm not only finds this output but does so efficiently. This efficiency arises because \alggs searches for bounds jointly across all template directions using gradients, while the LP-solver-based transformer solves a separate optimization problem per direction. Figure~\ref{fig:solver-time} further confirms this, with the LP-based transformer's runtime growing steeply with the number of directions, while \toolname’s increases much more gradually.

\begin{figure}[t]
    \centering
    \begin{subfigure}[b]{0.32\linewidth}
        \centering
        \includegraphics[width=\linewidth]{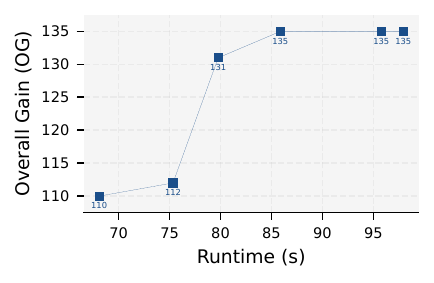}
        \vspace{0.25em}
        \includegraphics[width=0.9\linewidth]{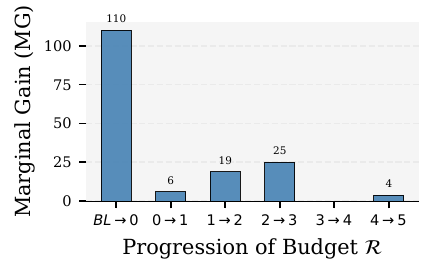}
        \caption{Zones}
        \label{fig:zones-quad}
    \end{subfigure}
    \hfill
    \begin{subfigure}[b]{0.32\linewidth}
        \centering
        \includegraphics[width=\linewidth]{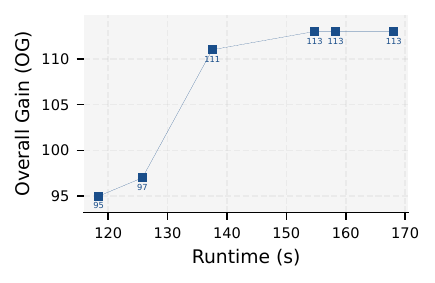}
        \vspace{0.25em}
        \includegraphics[width=0.9\linewidth]{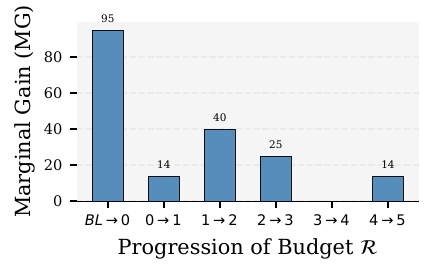}
        \caption{Octagons}
        \label{fig:oct-quad}
    \end{subfigure}
    \hfill
    \begin{subfigure}[b]{0.32\linewidth}
        \centering
        \includegraphics[width=\linewidth]{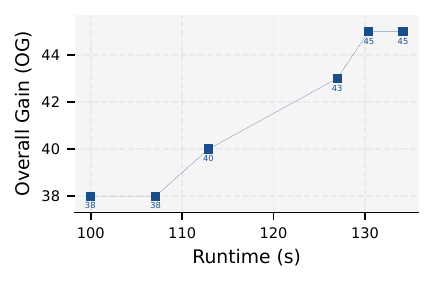}
        \vspace{0.25em}
        \includegraphics[width=0.9\linewidth]{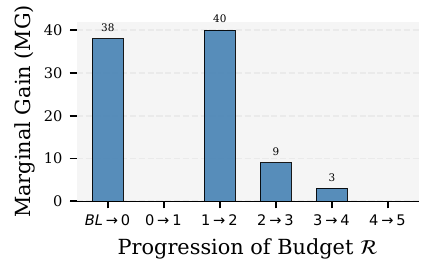}
        \caption{Polyhedra}
        \label{fig:pk-quad}
    \end{subfigure}

    \caption{Evolution of invariant strengthening while handling quadratic assignments across domains.}
    \label{fig:domains-quad}
\end{figure}

\subsection{Extending Analysis to Nonlinear Operators and Expressive Domains}
\label{sec:eval2}

\pheader{Quadratic Assignment Sequences}
Unlike in Section~\ref{sec:eval1}, where \toolname was restricted to affine assignments for ground-truth comparison, we now use it in its full setting, allowing it to apply our evolving transformer to quadratic assignments and sequences as well, as discussed in Sec.~\ref{sec:framework}. Since solvers like Gurobi are unsound for nonlinear operators~\cite{gurobi_nonlinear_unsound} and Symba~\cite{awsSolver} supports only linear arithmetic, obtaining ground-truth outputs is infeasible in this setting. As before, we run \toolname with gradient-step budgets varying from $0$ to $5$, but track only OG and MG metrics, as RG cannot be measured without ground truth. The increasing OG values in Figs.~\ref{fig:zones-quad} and~\ref{fig:oct-quad}, and the non-zero MG values ($4 \rightarrow 5$) even when OG plateaus, show that \toolname continues to strengthen invariants as the budget increases, enabling adaptable analysis even for quadratic operators. Overall gain (OG) reaches 135 for Zones and 113 for Octagons, compared to 121 and 59 in the linear case, indicating that handling quadratic operators further strengthens invariants and yields more precise analyses.

\pheader{Polyhedra Domain} Since Polyhedra is not a template-constraint (TCM) domain and supports arbitrary linear constraints, we configure \toolname to use the Zones~\cite{zones} template that tracks variable differences (\toolname\ always generates outputs within a chosen template, Sec.~\ref{sec:upose}). As Polyhedra can already represent linear updates exactly, \toolname applies the evolving transformer only to quadratic assignments and assignment sequences with overall quadratic effects. Due to numerical overflows in \elina specific to the Polyhedra domain, we run experiments on a subset of 49 benchmarks (details in Appendix ~\hyperref[sec:appeval2]{D.2}), yielding a baseline invariant set $\invs_{bl}$ of size 195. The results in Fig.~\ref{fig:pk-quad} show that \toolname enables precise and adaptable analysis even for expressive domains like Polyhedra, strengthening 45 (about 25\%) of the baseline invariants by $\mathcal{R} = 5$. demonstrating the benefits of precise (and sequence-level) reasoning for quadratic operators.

\pheader{Runtime Trends Across Domains} 
As seen in Fig.~\ref{fig:domains-quad}, the analysis time does not grow linearly with the gradient-step budget $\mathcal{R}$. This is expected, since $\mathcal{R}$ controls only the transformer outputs, while the total runtime depends on the full sequence of all transformer calls during the analysis. Moreover, as the output space is complex, the search may get stuck in local minima and produce identical outputs across consecutive budget settings. Nevertheless, our gradient-guided traversal is generally effective at escaping such plateaus. For instance, in Zones and Octagons (Figs.~\ref{fig:zones-quad} and~\ref{fig:oct-quad}), budgets $\mathcal{R}{=}3$ and $\mathcal{R}{=}4$ yield identical results, but precision improves at $\mathcal{R}{=}5$, as indicated by the non-zero MG from $\mathcal{R}{=}4$ to $\mathcal{R}{=}5$. A similar pattern appears for Polyhedra (Fig.~\ref{fig:pk-quad}), where progress resumes after stagnation between $\mathcal{R}{=}0$ and $\mathcal{R}{=}1$.

\pheader{Efficiency Comparison with \elina}
Analyzing the benchmarks with \elina takes approximately 20 seconds (18.53s for Zones, 20.11s for Octagons, and 23.2s for Polyhedra). In contrast, as shown above, \toolname incurs higher runtime, ranging from 70 to 170 seconds as the number of gradient steps increases from 0 to 5. This overhead is expected, as \toolname consistently computes more precise invariants than \elina. Importantly, \toolname enables users to navigate the precision--runtime tradeoff by selecting configurations tailored to their needs, and as discussed in the last section, it even outperforms the industrial-grade, heavily optimized LP solver Gurobi based transformer in computing the most precise invariants. Moreover, for a fair comparison on standard benchmarks, we implement \toolname within the \clam/\elina framework, incurring integration overhead due to conversions between ELINA’s internal abstract states and the tensor-based representation used for gradient-based optimization ($\approx$50\% of runtime overhead). This cost is not inherent and can be reduced with additional engineering effort by implementing a native abstract domain tailored to \toolname-style parametric transformers.

In terms of memory, \toolname introduces only modest overhead, as it does not explicitly compute all abstract outputs, but instead efficiently represents the space of sound abstract outputs \textit{symbolically} via parameters. In our evaluation, the median memory usage increases from 0.08\,GB to 0.14\,GB (max 0.23\,GB), corresponding to well under 1\% of memory on our 32\,GB machine ($\approx$0.7\% at peak), making the footprint negligible. This increase is expected, as our PyTorch~\cite{pytorch}-based gradient descent relies on tensor computations that consume some additional memory.

\pheader{Ineffectiveness of Random Sampling}
The results in Fig.~\ref{fig:zones-oct-linear} and Fig.~\ref{fig:domains-quad} show the efficacy of our AGG search in discovering more precise invariants as the budget $\mathcal{R}$ increases. As discussed in Sec.~\ref{sec:agg}, a naive search strategy can be to randomly sample parameters from $\Theta$. We evaluate this using rejection sampling, where we aim to obtain $r$ feasible points for budget $\mathcal{R} = r$, ideally drawing samples for each point until a feasible one is found. However, since $\Theta$ is a high-dimensional polyhedron defined by complex linear constraints, most randomly drawn points are either infeasible or far from the high-scoring regions where precise outputs lie, and satisfying these constraints by chance could require arbitrarily many attempts. To match AGG's runtime, we cap this at 5 attempts per point, accepting the first feasible sample found and falling back to the interval relaxation ($\lambda = 0$) otherwise. This approach yields \textit{zero improvement} over $\mathcal{R} = 0$, confirming that unguided search fails regardless of budget. This highlights the importance of our adaptive gradient-guided search, which jointly manages feasibility and objective-awareness by adaptively switching between optimizing $\mathcal{J}$ when feasible and restoring feasibility otherwise, efficiently navigating $\Theta$ toward precise outputs even with small budgets.

\pheader{Precision Benefits from Sequence-Level Reasoning}
The results in Fig.~\ref{fig:domains-quad} use our standard sequence-level
analysis, where the instruction sequences supported by \toolname are analyzed
jointly. This joint handling avoids the precision loss incurred when composing
per-instruction transformers and is a key contributor to the gains we observe.
As shown in \hyperref[sec:appeval3]{Appendix ~D.3}, disabling sequence-level
handling leads to substantially fewer strengthened invariants.
\section{Related Works}

Manually designing efficient abstract transformers is both challenging and tedious and has motivated substantial work on automating their construction~\cite{autoabs, abstractionSyn1, abstractionSyn2, datalyzer, pasado, Linsyn, nai}. Synthesis-based approaches such as \cite{abstractionSyn1, abstractionSyn2} use program synthesis over a user-provided DSL to generate transformers for specific domains such as strings.
Data-driven techniques like \cite{pasado, Linsyn} train neural networks to serve as abstract transformers for specific non-linear operations, but these methods remain limited to bounded abstract elements.
Neither of these directions applies to expressive relational domains like Octagons~\cite{octagon} and Polyhedra~\cite{polyhedra}, where abstract elements may be unbounded and involve many numeric constraints, making properties like soundness difficult to enforce. This is evident in \cite{nai}, which attempts to learn transformers for such domains, but the trained models have no soundness guarantees, making them unsuitable for static analysis. Another line of work~\cite{awsSolver,blockTransformer,effsymbabs,posthat} uses symbolic abstraction~\cite{symbopt} to compute most-precise transformers for these domains, not only for individual instructions but also for entire sequences, thus enabling more precise analysis than per-instruction reasoning. However, these methods typically rely on expensive solver calls, which limits their scalability. Our evolving abstract transformer addresses these limitations, as it is general and not tied to any particular domain or instruction. The same transformer works for all supported domains and operators, while also enabling efficient sequence-level reasoning.

While symbolic optimization yields precise transformers and works such as~\cite{effsymbabs} improve its efficiency, it remains solver-dependent and not general across domains, causing libraries to default to faster but less precise transformers regardless of the analysis task, which is limiting. Prior work~\cite{tune1,tune2, symtuner} has demonstrated the benefits of tuning general analyzer parameters for specific downstream goals, and~\cite{bai,resourceaware} show the value of adapting analyses to time and memory budgets of the task. However, these approaches act only at the analyzer level and do not provide fine-grained control over the precision–efficiency trade-offs of individual transformers for varied scenarios. Our evolving abstract transformer addresses this by moving from computing a single fixed output to searching within a parametric space of sound outputs using an efficient gradient descent-based method that adapts to the task. For abstract interpretation, to the best of our knowledge, gradient-based methods have been used only in neural network verification settings~\cite{wang2021betacrown, xufast, ferrari2022complete, racoon}, where they are applied to hand-crafted parametric transformers for specific operators such as ReLU over bounded interval domains. These settings are simpler, and extending gradient-based methods to complex unbounded domains is substantially more difficult. Our formulation achieves this by transforming efficient yet unreliable first-order optimization methods such as gradient descent into a reliable mechanism with sound-by-construction guarantees, by using them within a parametric space of sound outputs to identify more precise outputs for instructions and instruction sequences.

\section{Conclusion}
We introduce the \toolname{} framework, built on our Evolving Abstract Transformer, addressing the fundamental limitation of existing transformers that have fixed imprecision, are non-adaptable, and must be manually designed per domain and operator. \toolname{} uses our UPOSE algorithm to generate a parametric space of sound outputs for a broad class of domains and operators, and our AGG algorithm to efficiently search this space for outputs best suited to the analysis objectives. This reformulation into an efficient search problem, combined with gradient-based optimization while preserving soundness, enables precise and adaptable analysis where the number of gradient steps can be adjusted to flexibly trade precision for efficiency. Across domains, \toolname{} achieves up to 3.2$\times$ faster convergence to the most precise invariants compared to existing baselines. This work opens up promising directions for accelerator-driven program analysis on modern hardware such as GPUs, as well as learning-based approaches to abstract interpretation.

\clearpage
\section*{ACKNOWLEDGEMENTS}

We thank the anonymous reviewers for their insightful comments, which greatly helped us improve the clarity and presentation of this work. This work was supported in part by a grant from the Amazon Illinois Center on AI for Interactive Conversational Experiences (AICE), NSF Grants No. CCF-2238079, CCF-2316233, CNS-2148583, and NAIRR240476, an Open Philanthropy research grant, and compute resources from Indiana Jetstream2.

\section*{DATA AVAILABILITY STATEMENT}

The source code and evaluation data used for the experiments in this paper are publicly available on Zenodo~\cite{absevolve_artifact} (DOI: \href{https://doi.org/10.5281/zenodo.19079444}{10.5281/zenodo.19079444}), with instructions to reproduce the results. The latest version of the framework is maintained on GitHub at \href{https://github.com/uiuc-focal-lab/AbsEvolve}{https://github.com/uiuc-focal-lab/AbsEvolve}.
\bibliographystyle{ACM-Reference-Format}
\bibliography{sections/bibliography}
\clearpage
\appendix
\section{Duality-based construction of Sound Parametric Scalar Map}
\label{app:derivation}
%
%

Our algorithm to generate a parametric family of sound transformers relies on finding a sound Parametric Scalar Map (PSM) (Def~\ref{def:parameter-scalar-map}) for the optimal value $o_{\min}$ of optimization problems of the form:
\begin{equation}
o_{\min} \;=\; 
\min_{\mathbf{v} \in \mathbb{R}^n}\; f(\mathbf{v})
\quad \text{s.t.}\quad A\,\mathbf{v} \le \mathbf{b}\,
\label{eq:omin}
\end{equation}
where the objective function $f(\mathbf{v})$ is of the form:
\begin{equation}
f(\mathbf{v}) = 
\sum_{1 \le i < k \le n} H_{ik}\, v_i v_k
+ \sum_{i=1}^{n} Q_i\, v_i^2
+ \mathbf{c}^{\top}\mathbf{v}
+ d
\label{eq:fv}
\end{equation}
Here, $H_{ik}\, v_i v_k$ denotes the bilinear terms (products of distinct variables), $Q_i\, v_i^2$ denotes the quadratic terms, $c_i\, v_i$ denotes the linear terms, and $d \in \mathbb{R}$ is a constant. 

\paragraph{Soundness criterion.}
A scalar $\ell$ is a \emph{sound lower bound} on the optimal value $o_{\min}$ whenever $\ell \le o_{\min}$.

We obtain a sound Parametric Scalar Map for Eq~\ref{eq:omin} through the following steps.

\subsection{Step 1: Efficient Dual Construction}
\label{subsec:eff-dual}

As our goal is to obtain certified \emph{lower bounds} on the primal optimum $o_{\min}$ for the problem in Eq~\ref{eq:omin}, we turn to its Lagrangian dual, because weak duality guarantees that any feasible value of the dual function is a lower bound. Introducing non-negative multipliers $\boldsymbol{\lambda}\!\ge 0$ for the constraints $A\mathbf{v}\le\mathbf{b}$ in the problem in Eq~\ref{eq:omin} yields the following re-formulation of the same problem:
\[
  o_{\min}\;=\;
  \min_{\mathbf{v}\in\mathbb{R}^n}\,
  \max_{\boldsymbol{\lambda}\ge 0}
  \bigl(f(\mathbf{v})+\boldsymbol{\lambda}^{\top}(A\mathbf{v}-\mathbf{b})\bigr)
\]

Swapping the order of optimization produces the dual problem
\[
  o_{\text{dual}}\;=\;
  \max_{\boldsymbol{\lambda}\ge 0}\,
  \underbrace{\min_{\mathbf{v}\in\mathbb{R}^n}
    \bigl(f(\mathbf{v})+\boldsymbol{\lambda}^{\top}(A\mathbf{v}-\mathbf{b})\bigr)}
    _{g(\boldsymbol{\lambda})}
\]
where $g(\boldsymbol{\lambda})$ is the \emph{dual function}.  By weak duality
\[
  g(\boldsymbol{\lambda}) \;\le\; o_{\min}\quad\text{for every } \boldsymbol{\lambda}\ge 0
\]
Consequently, if we can compute the inner minimization $g(\boldsymbol{\lambda})$ for a given $\boldsymbol{\lambda}\!\ge 0$, then inserting \emph{any} such multiplier values produces a sound lower bound, thus giving us the desired parametric family of bounds.

\paragraph{Inefficiency of the naive dual construction.}
Before evaluating the dual function $g(\boldsymbol{\lambda})$, note that
constructing the dual from the entire constraint set $A\mathbf{v}\le\mathbf{b}$
is unnecessarily costly.  The abstract input element—like most numerical
abstract domains—already contains tight \emph{box constraints}
$l_i\le v_i\le u_i$.  Including these bounds in the dual construction
introduces two Lagrange multipliers per variable, inflating the dual vector
and increasing the cost of every evaluation of $g$. A standard remedy is to \emph{peel off} the box constraints first: one rewrites the feasible region as \(\mathbf{v}\in[l,u]\) together with the remaining
coupled inequalities, and constructs the dual only for the latter.  The inner minimisation then becomes a box-constrained problem, which typically admits a closed-form solution or a very cheap projection step~\cite[Ch.~1]{bertsekas1999nonlinear};\, \cite[Ch.~17]{nocedal2006numerical}.  This strategy keeps the dual compact, reduces the cost of computing \(g(\boldsymbol{\lambda})\), and accelerates the subsequent maximisation over \(\boldsymbol{\lambda}\!\ge 0\).

So, we isolate tight lower and upper bounds for each variable from the constraint system \(A\mathbf{v}\le\mathbf{b}\).  This yields a bounding box \(\mathbf{v}\in[\mathbf{l},\mathbf{u}]\), where \(l_i\le v_i\le u_i\) for each \(i\).  (Some bounds may be unbounded, i.e.\ \(l_i=-\infty\) or \(u_i=+\infty\), depending on the original constraints.)  Removing these trivial bounds from the coupled inequalities leaves a reduced system
\(A'\mathbf{v}\le\mathbf{b}'\).  The optimization problem is therefore equivalent to
\begin{equation}
  o_{\min}\;=\;
  \min_{\mathbf{v}\in[\mathbf{l},\mathbf{u}]}\;
  f(\mathbf{v})
  \quad\text{s.t.}\quad
  A'\mathbf{v}\le\mathbf{b}'
  \label{eq:omin-bounded}
\end{equation}
where the box constraint is handled explicitly and only the coupled inequalities are retained in \(A'\mathbf{v}\le\mathbf{b}'\). Introducing non-negative multipliers $\boldsymbol{\lambda}\!\ge 0$ for the constraints $A'\mathbf{v}\le\mathbf{b'}$ in the problem in Eq~\ref{eq:omin-bounded} yields the following re-formulation:
\begin{equation}
  o_{\min}\;=\;
  \min_{\mathbf{v}\in[\mathbf{l},\mathbf{u}]}\;
  \max_{\boldsymbol{\lambda}\ge 0}
  \bigl(f(\mathbf{v})+\boldsymbol{\lambda}^{\top}(A'\mathbf{v}-\mathbf{b}')\bigr)
  \label{eq:omin-bounded-lag}
\end{equation}

Swapping the order of optimization then produces the dual problem
\begin{equation}
  o_{\text{dual}}\;=\;
  \max_{\boldsymbol{\lambda}\ge 0}\,
  \underbrace{\min_{\mathbf{v}\in[\mathbf{l},\mathbf{u}]}
    \bigl(f(\mathbf{v})+\boldsymbol{\lambda}^{\top}(A'\mathbf{v}-\mathbf{b}')\bigr)}
    _{g(\boldsymbol{\lambda})}
  \label{eq:omin-bounded-dual}
\end{equation}
where $g(\boldsymbol{\lambda})$ is the \emph{dual function}. Now, as discussed above, by weak duality, every value $g(\boldsymbol{\lambda})$ with $\boldsymbol{\lambda}\ge 0$ is a certified lower bound on $o_{\min}$. Henceforth, we focus on evaluating this dual function efficiently: once we can compute $g(\boldsymbol{\lambda})$ for any chosen multiplier, we immediately obtain an entire parametric family of sound bounds.

Substituting $f(\mathbf{v})$ from Eq.~\eqref{eq:fv}:
\[
g(\boldsymbol{\lambda}) \;=\;
\min_{\mathbf{v}\in[\mathbf{l},\mathbf{u}]}\;\left(
\sum_{1 \le i < k \le n} H_{ik}\, v_i v_k
\;+\; \sum_{i=1}^n Q_i\, v_i^2
\;+\; (\mathbf{c} + A'^T\lambda)^T \mathbf{v}
\;+\; d - \lambda^T \mathbf{b}'
\right)\
\]

Note, however, that this inner problem is still far from trivial: even with the simple box constraint, the presence of arbitrary bilinear terms $H_{ik}v_i v_k$ and potentially indefinite quadratic coefficients $Q_i$ turns the minimisation into a general (and possibly non-convex) quadratic programme, so an exact closed-form solution for $g(\boldsymbol{\lambda})$ is not available in general.

\subsection{Step 2: Sound inner minimisation via linear coefficient splitting}
\label{subsec:inner-min}

We now discuss how to make the inner minimisation problem \emph{tractable}.  The idea is to redistribute each linear coefficient between its associated quadratic and bilinear terms so that the objective can be broken into one– and two–variable “boxes” that can be optimised independently.
    
\medskip
\noindent
\textbf{Step 1: Introduce splitting parameters.}  
Pick auxiliary vectors $\mathbf{S}\in\mathbb{R}^{n}$ and
$\mathbf{D}\in\mathbb{R}^{n\times n}$ and rewrite
$g(\boldsymbol{\lambda})$ as
\[
\begin{aligned}
g(\boldsymbol{\lambda}) =  g(\boldsymbol{\lambda},\mathbf{S},\mathbf{D}) =
\min_{\mathbf{v}\in[\mathbf{l},\mathbf{u}]}\Bigl(
      &\sum_{1\le i<k\le n}
         \bigl(H_{ik}v_i v_k + D_{ik}^{i}v_i + D_{ik}^{k}v_k\bigr)\\
      &+\sum_{i=1}^{n}\bigl(Q_i v_i^{2} + S_i v_i\bigr)\\
      &+\bigl(\mathbf{c}+A^{\top}\lambda-\mathbf{S}-\boldsymbol{\Delta}\bigr)^{\!\top}\mathbf{v}
       +d-\lambda^{\top}\mathbf{b}\Bigr),
\end{aligned}
\]
where $\Delta_i:=\sum_{k\neq i}D_{ik}^{i}$ and
$\boldsymbol{\Delta}:=(\Delta_1,\dots,\Delta_n)^{\top}$.
For any fixed $(\mathbf{S},\mathbf{D})$ this is algebraically identical to the
original $g(\boldsymbol{\lambda})$; only the grouping of the linear terms has changed.

\medskip
\noindent
\textbf{Step 2: Minimise each box separately.}  
Splitting the objective term-wise yields the relaxed value
\begin{align}
\hat g(\boldsymbol{\lambda},\mathbf{S},\mathbf{D})=\;
&\sum_{1\le i<k\le n}
  \min_{(v_i,v_k)\in[l_i,u_i]\times[l_k,u_k]}
       \bigl(H_{ik}v_i v_k + D_{ik}^{i}v_i + D_{ik}^{k}v_k\bigr) \notag\\
&+\sum_{i=1}^{n}
  \min_{v_i\in[l_i,u_i]}
       \bigl(Q_i v_i^{2} + S_i v_i\bigr) \notag\\
&+\min_{\mathbf{v}\in[\mathbf{l},\mathbf{u}]}
       \bigl(\mathbf{c}+A^{\top}\lambda-\mathbf{S}-\boldsymbol{\Delta}\bigr)^{\!\top}\mathbf{v}
       +d-\lambda^{\top}\mathbf{b}.
\label{eq:relax}
\end{align}

\medskip
\noindent
\textbf{Step 3: Soundness of the relaxation.}  
For any family $\{h_j\}$ on a common domain $\mathcal{X}$,
\begin{equation}
\min_{x\in\mathcal{X}}\sum_{j}h_j(x)\;\ge\;\sum_{j}\min_{x\in\mathcal{X}}h_j(x),
\label{eq:min-split}
\end{equation}
so applying \eqref{eq:min-split} to \eqref{eq:relax} gives
\begin{equation}
  \hat g(\boldsymbol{\lambda},\mathbf{S},\mathbf{D})\;\le\; g(\boldsymbol{\lambda})\;\le\; o_{\min},
\qquad
\text{for all }\boldsymbol{\lambda}\ge 0,\;\mathbf{S},\;\mathbf{D}.
\label{eq:deriv-sound}
\end{equation}

\medskip
\noindent
\textbf{Intuition.}  
The split yields $n(n\!-\!1)/2$ two-variable boxes, $n$ single-variable boxes, and one residual linear box.  
Each box now contains only a single quadratic, bilinear, or linear term, so every sub-problem is straightforward to minimise.  
Summing their minima gives the relaxed value $\hat g(\boldsymbol{\lambda},\mathbf{S},\mathbf{D})$, which remains a sound lower bound as discussed above.

\subsection{Step 3: Efficient Evaluation of Bilinear, Quadratic, and Linear Terms}
\label{subsec:inner-parts-eval}

We have established that $\hat g(\boldsymbol{\lambda},\mathbf{S},\mathbf{D})$ is a sound lower bound for every $(\boldsymbol{\lambda},\mathbf{S},\mathbf{D})$ with $\boldsymbol{\lambda}\ge 0$, regardless of the particular choices of $\mathbf{S}$ and $\mathbf{D}$.  
Equivalently, we can view $(\boldsymbol{\lambda},\mathbf{S},\mathbf{D})$ as a parameter vector and $\Theta$ as the feasible parameter space
\[
\Theta \;=\;\{\;(\boldsymbol{\lambda},\mathbf{S},\mathbf{D})\in\mathbb{R}^m\times\mathbb{R}^n\times\mathbb{R}^{n\times n}\;\mid\;\boldsymbol{\lambda}\ge 0\;\}\,,
\]
so that for every parameter in $\Theta$ we obtain a finite, certified bound $\hat g(\boldsymbol{\lambda},\mathbf{S},\mathbf{D})\le o_{\min}$.

Next, we describe how to \emph{compute} $\hat g(\boldsymbol{\lambda},\mathbf{S},\mathbf{D})$ efficiently.  
For fixed $(\boldsymbol{\lambda},\mathbf{S},\mathbf{D})$, the relaxed objective decomposes into independent one- and two-dimensional subproblems.  
By solving each subproblem separately and summing their optimal values, we obtain the value of $\hat g(\boldsymbol{\lambda},\mathbf{S},\mathbf{D})$ and, at the same time, we update the feasible parameter space $\Theta$ to reflect any additional constraints discovered during this process. More concretely, the value of $\hat g(\boldsymbol{\lambda},\mathbf{S},\mathbf{D})$ is obtained by solving the following one- and two-dimensional subproblems for each class of terms:

\begin{enumerate}
    \item \textbf{Bilinear terms:} Consider a general bilinear subproblem (assume $A \neq 0$ as only such terms occur when we split the objective):
    \[
    \min_{(v_i,v_k)\in[l_i,u_i]\times[l_k,u_k]}\Big(A\, v_i v_k + B\, v_i + C\, v_k\Big)\,.
    \]
    
    If we look at $\hat g(\boldsymbol{\lambda},\mathbf{S},\mathbf{D})$ in Eq~\eqref{eq:relax}, we see that $B$ and $C$ depend on the parameters $\mathbf{D}$ introduced in our splitting (and $A$ is a fixed constant), so the conditions derived below give rise to corresponding constraints that we add to the parameter set $\Theta$:
    \[
    \Theta \;\leftarrow\; \Theta \;\cap\; \{\text{new sign or bound constraint}\}\,.
    \]
    
    First, we rewrite:
    \[
    A\, v_i v_k + B\, v_i + C\, v_k \;=\; \big(A\, v_i + C\big)\Big(v_k+\frac{B}{A}\Big)-\frac{B C}{A}\,,
    \]
    so that the product is between two affine terms with intervals:
    \[
    A\, v_i+C \in \begin{cases}
    [A\, l_i+C,\;A\, u_i+C]\,, & A\ge 0\,,\\[0.3em]
    [A\, u_i+C,\;A\, l_i+C]\,, & A < 0\,,
    \end{cases}
    \qquad
    v_k+\frac{B}{A}\in[l_k+B/A,\;u_k+B/A]\,.
    \]
    
    Setting
    \[
    l_1 := \min(A\, l_i+C,\;A\, u_i+C),\qquad u_1 := \max(A\, l_i+C,\;A\, u_i+C),
    \]
    \[
    l_2 := l_k+\frac{B}{A},\qquad u_2 := u_k+\frac{B}{A}\,,
    \]
    the value extreme products are $\min(\{\,l_1 l_2,\;l_1 u_2,\;u_1 l_2,\;u_1 u_2\,\}) - BC/A$.
    
    There are four configurations in which one of these corner-products may diverge to $-\infty$:
    \begin{enumerate}
    \item $l_1 = -\infty$ and $u_2 > 0$,
    \item $l_2 = -\infty$ and $u_1 > 0$,
    \item $u_1 = +\infty$ and $l_2 < 0$,
    \item $u_2 = +\infty$ and $l_1 < 0$.
    \end{enumerate}
    If in any of these cases the corresponding paired bound is also unbounded (e.g. $u_2 = +\infty$, $u_1 = +\infty$, $l_2 = -\infty$, or $l_1 = -\infty$, respectively), then the objective is unbounded and cannot yield a finite lower bound. In this situation we add the contradictory constraint
    \[
    \Theta \;\leftarrow\; \Theta \;\cap\; \{\,0 \le -1\,\}\,,
    \]
    rendering the parameter set infeasible so that the resulting subproblem value is $-\infty$.
    
    If, instead, the paired bound is finite, we add the appropriate sign constraint ($u_2 \le 0$, $u_1 \le 0$, $l_2 \ge 0$, or $l_1 \ge 0$) to $\Theta$:
    \[
    \Theta \;\leftarrow\; \Theta \;\cap\; \{\text{the corresponding bound constraint}\}\,,
    \]
    so that the problematic product cannot become negatively unbounded.
    
    Once all such constraints have been added to $\Theta$, the bilinear subproblem is well-posed. Its optimal value is then
    \[
    l_{\mathrm{bilin}} \;:=\; \min\{\,l_1 l_2,\;l_1 u_2,\;u_1 l_2,\;u_1 u_2\,\}\;-\;\frac{B C}{A}\,,
    \]
    which is finite if the constraints are satisfied. Otherwise, if the constraints cannot be satisfied, this subproblem value is $-\infty$.

    \item \textbf{Quadratic terms:} Consider a general quadratic subproblem (assume $A \neq 0$ as only such terms occur after splitting the objective):
    \[
    \min_{v_i\in[l_i,u_i]} \;\Big(A\, v_i^2 + B\, v_i\Big)\,,
    \]
    where $B$ depends on the splitting parameters~$\mathbf{S}$, so the conditions derived below give rise to corresponding constraints that we add to~$\Theta$:
    \[
    \Theta \;\leftarrow\; \Theta \;\cap\; \{\text{new constraint}\}\,.
    \]
    
    \begin{itemize}
    \item If $A>0$ (convex), we first compute the unconstrained minimizer
    \[
    v_i^* := -\frac{B}{2A}\,.
    \]
    $v_i^*$ is the global minimizer. If $v_i^*\in[l_i,u_i]$, the objective value is
    \[
    l_{\mathrm{quad}} := A\,(v_i^*)^2 + B\, v_i^*\,,
    \]
    which is finite.  If $v_i^*\notin[l_i,u_i]$, then the minimizer lies at one of the endpoints.  Moreover, the unconstrained minimizer cannot lie outside the box if both bounds were truly unbounded (i.e. $l_i=-\infty$ and $u_i=+\infty$), so at least one of the endpoints is finite. Thus, we can always evaluate the objective at the finite endpoint(s) and take the smaller value.
    \[
    l_{\mathrm{quad}} := \min\{\,A\, l_i^2 + B\, l_i,\;A\, u_i^2 + B\, u_i\,\}\,,
    \]
    which is finite.
    
    \item If $A<0$ (concave), then the objective can decrease without bound unless both $l_i$ and $u_i$ are finite.  
    More precisely:
    \begin{itemize}
    \item If $l_i = -\infty$, then as $v_i \to -\infty$ we have $A\, v_i^2 \to -\infty$.
    \item If $u_i = +\infty$, then as $v_i \to +\infty$ we have $A\, v_i^2 \to -\infty$.
    \end{itemize}
    In either of these cases the subproblem is unbounded below, so we add the contradictory constraint
    \[
    \Theta \;\leftarrow\; \Theta \;\cap\; \{\,0 \le -1\,\}\,,
    \]
    rendering the parameter set infeasible and the subproblem value $-\infty$.  
    If both $l_i$ and $u_i$ are finite, we evaluate at the two endpoints and take
    \[
    l_{\mathrm{quad}} := \min\{\,A\, l_i^2 + B\, l_i,\;A\, u_i^2 + B\, u_i\,\}\,,
    \]
    which is finite.
    
    \end{itemize}

    \item \textbf{Linear terms:} Consider the subproblem
    \[
    \min_{\mathbf{v} \in [\mathbf{l}, \mathbf{u}]} \; \mathbf{c}^\top \mathbf{v} + d
    \]
    We compute the value of this objective variable-wise. That is, for each term~$c_i \cdot v_i$ where $v_i \in [l_i, u_i]$, we compute its minimum contribution~$L_i$ and determine a corresponding set of constraints~$\Theta_i$ on the parameters (i.e., on the coefficients~$c_i$) that ensure this minimum is finite. We consider four exhaustive cases:
    \begin{enumerate}[leftmargin=*]
        \item \textit{Both bounds finite:} The minimum depends on the sign of~$c_i$: if $c_i \geq 0$, the minimum is $c_i \cdot l_i$; if $c_i < 0$, it is $c_i \cdot u_i$. This can be expressed uniformly as:
        \[
        L_i = -\frac{|c_i|}{2}(u_i - l_i) + \frac{c_i}{2}(u_i + l_i)
        \]
        No constraint is needed on~$c_i$ in this case.
    
        \item \textit{Lower bound is $-\infty$, upper bound is finite:} If $c_i > 0$, the minimum is unbounded ($-\infty$). To ensure a finite value, we require $c_i \leq 0$, in which case $L_i = c_i \cdot u_i$.
    
        \item \textit{Lower bound is finite, upper bound is $+\infty$:} If $c_i < 0$, the minimum is unbounded. To ensure a finite value, we require $c_i \geq 0$, in which case $L_i = c_i \cdot l_i$.
    
        \item \textit{Both bounds infinite:} In this case, the minimum is finite only if $c_i = 0$, so we add this constraint and take $L_i = 0$.
    \end{enumerate}

We apply this procedure to all coordinates $i = 1, \dots, n$, summing the individual contributions to compute the overall minimized value:
\[
L = \sum_{i=1}^{n} L_i + d
\]
The complete constraint set~$\Theta$ is formed by collecting the individual~$\Theta_i$ across all variables. If the parameters $(\boldsymbol{\lambda}, \mathbf{S}, \mathbf{D})$ satisfy~$\Theta$, then the dual value is well-defined and finite. Otherwise, the objective is unbounded and evaluates to~$-\infty$.

\end{enumerate}

\noindent\textbf{Conclusion.}\enspace
In Step~\ref{subsec:inner-min}, we constructed the relaxed objective $\hat g(\boldsymbol{\lambda},\mathbf{S},\mathbf{D})$ by decomposing the inner minimization into independent subproblems and applying the min-split inequality. This yielded a symbolic under-approximation of~$o_{\min}$. In Step~\ref{subsec:inner-parts-eval}, we showed how to efficiently evaluate $\hat g(\boldsymbol{\lambda},\mathbf{S},\mathbf{D})$ by solving each bilinear, quadratic, and linear subproblem separately and aggregating the resulting sign and bound constraints into a parameter set~$\Theta$.

Together, these steps define a \emph{Parametric Scalar Map} $(\Theta, L)$ for the optimization problem~\eqref{eq:omin}, where
\[
L:\Theta \to \mathbb{R}, \qquad L(\boldsymbol{\lambda}, \mathbf{S}, \mathbf{D}) := \hat g(\boldsymbol{\lambda}, \mathbf{S}, \mathbf{D})\,.
\]
By construction, each $(\boldsymbol{\lambda}, \mathbf{S}, \mathbf{D}) \in \Theta$ yields a finite, sound lower bound on~$o_{\min}$—that is,
\[
L(\boldsymbol{\lambda}, \mathbf{S}, \mathbf{D}) \le o_{\min}
\quad \text{and} \quad
L(\boldsymbol{\lambda}, \mathbf{S}, \mathbf{D}) \in \mathbb{R}\,.
\]
If $\Theta$ is empty or infeasible, then $(\Theta, L)$ yields no informative bound beyond the trivial $-\infty$. This may occur if the true optimum is unbounded or if the relaxation is too coarse to derive a finite, sound lower bound. Thus, Steps~\ref{subsec:eff-dual}--\ref{subsec:inner-parts-eval} compute a sound PSM for the original optimization problem.

\pheader{Special case: linear objectives}
Our construction also applies when the objective function~$f(\mathbf{v})$ is purely linear, i.e., when $H = 0$ and $Q = 0$. In this setting, the primal optimization problem becomes a standard linear program:
\[
f(\mathbf{v}) = \mathbf{c}^\top \mathbf{v} + d, \qquad o_{\min} = \min_{\mathbf{v} \in [\mathbf{l}, \mathbf{u}]} \;\mathbf{c}^\top \mathbf{v} + d \quad \text{s.t. } A'\mathbf{v} \le \mathbf{b}'.
\]

Here, Step~2 (linear coefficient splitting) is no longer needed, as there are no bilinear or quadratic terms to redistribute. The Parametric Scalar Map simplifies to the pair~$(\Theta, L)$, where:
\[
\Theta = \{\; \boldsymbol{\lambda} \in \mathbb{R}_{\ge 0}^m \;\mid\; \text{linear subproblem yields finite bound} \;\}
\quad\text{and}\quad
L(\boldsymbol{\lambda}) = \hat g(\boldsymbol{\lambda})\,.
\]

However, \emph{feasibility constraints on~$\boldsymbol{\lambda}$ still arise} from the symbolic analysis of the linear evaluation step (see Section~\ref{subsec:inner-parts-eval}). Specifically, to ensure that the inner minimum over $[\mathbf{l}, \mathbf{u}]$ is finite, the affine objective coefficients $\mathbf{c} + A'^{\top} \boldsymbol{\lambda}$ must obey sign conditions based on unbounded directions of the box domain. For instance, if $l_i = -\infty$ for some coordinate $i$, then to avoid $-\infty$ in the objective, we must ensure that the $i$-th coefficient of $\mathbf{c} + A'^{\top} \boldsymbol{\lambda}$ is non-negative; similar conditions apply for $u_i = +\infty$. These constraints prune the parameter space~$\Theta$ accordingly. Thus, even in the linear case, our construction produces a valid Parametric Scalar Map that yields sound bounds~$L(\boldsymbol{\lambda}) \le o_{\min}$, with $\Theta$ encoding the feasibility region induced by symbolic soundness conditions.

\section{Algorithms}

\begin{algorithm}[h]
\caption{Compute Effective Update Map (EUM)}
\label{algo:eum}
\begin{algorithmic}[1]
\State \textbf{Input:} Instruction block $\mathcal{B} = [I_1, I_2, \dots, I_k]$; set of input variables $V$
\State \textbf{Output:} Map $\sigma$ from output variables to expressions over $V$
\Procedure{ComputeEUM}{$\mathcal{B}$, $V$}    
    \State Initialize $\sigma \gets \{\}$
    \For{each instruction $x := e$ in $B$}
    \State $\hat{e} \gets e[v \mapsto \sigma(v) \mid v \in \text{vars}(e) \cap \sigma]$ \Comment{with substitution and algebraic simplification}
    \State $\sigma(x) \gets \hat{e}$
    \EndFor
    \State \Return $\sigma$
\EndProcedure
\end{algorithmic}
\end{algorithm}

\begin{algorithm}[h]
\caption{Merging Assignment Sequences into QGO Blocks}
\label{algo:merge-blocks}
\begin{algorithmic}[1]
\State \textbf{Input:} CFG $\mathcal{F}$ with basic blocks $\{N_1, \ldots, N_k\}$
\State \textbf{Output:} Updated CFG $\widehat{\mathcal{F}}$ with merged instruction blocks; set of merged blocks $\mathbb{B}$
\Procedure{MergeAssignmentSequences}{$\mathcal{C}, \mathcal{V}$}
    \State Initialize $\mathbb{B} \gets \{\}$
    \For{each basic block $N$ in $\mathcal{C}$}
        \State Initialize $\mathcal{B}_N \gets [\,]$, $\texttt{curr} \gets [\,]$
        \For{each instruction $I$ in $N$}
            \If{$I$ is affine or quadratic assignment}
                \State $\texttt{tmp} \gets \texttt{curr} \cup [I]$
                \State $\sigma \gets$ \Call{ComputeEUM}{$\texttt{tmp}, \mathcal{V}$}
                \If{\Call{IsQuadraticBounded}{$\sigma$}}
                    \State $\texttt{curr} \gets \texttt{tmp}$
                \Else
                    \State Append $\texttt{curr}$ to $\mathcal{B}_N$; $\texttt{curr} \gets [I]$
                     \Comment{Degree exceeds 2 -> Instruction in new block}
                \EndIf
            \Else
                \If{$\texttt{curr} \neq [\,]$}
                    \State Append $\texttt{curr}$ to $\mathcal{B}_N$; $\texttt{curr} \gets [\,]$
                \EndIf
            \Comment{Skip non-assignment instruction $I$}
            \EndIf
        \EndFor
        \If{$\texttt{curr} \neq [\,]$}
            \State Append $\texttt{curr}$ to $\mathcal{B}_N$
        \EndIf
        \State Replace instructions of $N$ with $\mathcal{B}_N$; $\mathbb{B} \gets \mathbb{B} \cup \mathcal{B}_N$
    \EndFor
    \State \Return $(\widehat{\mathcal{C}}, \mathbb{B})$
\EndProcedure
\end{algorithmic}
\end{algorithm}

\medskip
\noindent
\textbf{Concrete Semantics of a Block.}\;
Let~$\seq = [\mathit{ins}_1; \dots; \mathit{ins}_n]$ be an admissible sequence of assignments, and let~$\block(\seq)$ be the corresponding block with Effective Update Map~$\sigma_{\block} : \mathcal{V} \to \mathbb{P}(\mathcal{V})$. Let~$\mathcal{A} = \{x_1, \dots, x_m\}$ be the set of variables updated in the block. Then, for any input state~$s$, the concrete big-step semantics of the block is defined as:
\[
\llbracket \block(\seq) \rrbracket(s) \triangleq 
s\big[x_1 \leftarrow \sigma_{\block}(x_1)(s),\;
      \dots,\;
      x_m \leftarrow \sigma_{\block}(x_m)(s)\big],
\]
where $\sigma_{\block}(x_i)(s)$ denotes the evaluation of the polynomial~$\sigma_{\block}(x_i)$ in the input state~$s$.


\section{Theorems and Proofs}

\begin{theorem}[Soundness of PSM Computing Procedure]
The procedure outlined in Appendix~\ref{app:derivation} computes a Parametric Scalar Map $\mathcal{M} = (\Theta, L)$ for optimization problems of the form:
\[
o_{\min} \;=\; 
\min_{\mathbf{v} \in \mathbb{R}^n}\; f(\mathbf{v})
\quad \text{s.t.}\quad A\,\mathbf{v} \le \mathbf{b},
\]
where the objective function $f(\mathbf{v})$ is a quadratic polynomial given by:
\[
f(\mathbf{v}) = 
\sum_{1 \le i < k \le n} H_{ik}\, v_i v_k
+ \sum_{i=1}^{n} Q_i\, v_i^2
+ \mathbf{c}^{\top}\mathbf{v}
+ d.
\]
The procedure is sound, meaning that for all $\theta \in \Theta$, the bound $L(\theta)$ satisfies $L(\theta) \le o_{\min}$.
\end{theorem}

\begin{proof}
As described in Appendix~\ref{app:derivation}, the procedure constructs the PSM $\mathcal{M} = (\Theta, L)$ using the following parameters:
\begin{itemize}
    \item Dual multipliers $\boldsymbol{\lambda} \in \mathbb{R}^m$ for the constraints $A\,\mathbf{v} \le \mathbf{b}$,
    \item Splitting parameters $\mathbf{S} \in \mathbb{R}^n$ and $\mathbf{D} \in \mathbb{R}^{n \times n}$ used to redistribute linear coefficients during the decomposition of the inner minimization.
\end{itemize}

The parameter space $\Theta$ consists of all tuples $(\boldsymbol{\lambda}, \mathbf{S}, \mathbf{D})$ that satisfy the feasibility constraints collected during the minimization of the decomposed objective components in Section~\ref{subsec:inner-parts-eval}, together with the constraint $\boldsymbol{\lambda} \ge 0$. The map $L$ is defined as:
\[
L(\boldsymbol{\lambda}, \mathbf{S}, \mathbf{D}) := \hat{g}(\boldsymbol{\lambda}, \mathbf{S}, \mathbf{D}),
\]
where $\hat{g}$ denotes the relaxed objective obtained by decomposing the dual objective $g(\boldsymbol{\lambda})$ into independent subproblems using the splitting parameters $\mathbf{S}$ and $\mathbf{D}$.

By the weak duality theorem, we know that:
\[
g(\boldsymbol{\lambda}) \;\le\; o_{\min} \qquad \text{for all } \boldsymbol{\lambda} \ge 0.
\]
Furthermore, by construction of the relaxation (Appendix~\ref{subsec:inner-min}), the decomposed bound satisfies:
\[
\hat{g}(\boldsymbol{\lambda}, \mathbf{S}, \mathbf{D}) \;\le\; g(\boldsymbol{\lambda}) \qquad \text{for all } \mathbf{S}, \mathbf{D}.
\]

Combining the two, we obtain:
\[
L(\boldsymbol{\lambda}, \mathbf{S}, \mathbf{D}) = \hat{g}(\boldsymbol{\lambda}, \mathbf{S}, \mathbf{D}) \;\le\; g(\boldsymbol{\lambda}) \;\le\; o_{\min},
\]
for all $\boldsymbol{\lambda} \ge 0$ and for all choices of $\mathbf{S}, \mathbf{D}$.

Since $\Theta$ is defined to include only those tuples $(\boldsymbol{\lambda}, \mathbf{S}, \mathbf{D})$ satisfying $\boldsymbol{\lambda} \ge 0$ along with the additional feasibility constraints, it follows that every $(\boldsymbol{\lambda}, \mathbf{S}, \mathbf{D}) \in \Theta$ satisfies the above inequality. Thus, $L$ yields a finite, sound lower bound on $o_{\min}$ for each valid parameter, establishing that $\mathcal{M} = (\Theta, L)$ is a sound Parametric Scalar Map for the given optimization problem. Note that the correctness of the proof relies on the correct computation of the relaxed bound $\hat{g}(\boldsymbol{\lambda}, \mathbf{S}, \mathbf{D})$, which relies on the correctness of the algebraic decompositions used to evaluate the independent bilinear, quadratic, and linear subproblems (Appendix~\ref{subsec:inner-parts-eval}). These decompositions involve standard algebraic manipulations and closed-form expressions, whose correctness we assume.

\label{thm:soundness-fspb}
\end{proof}

\begin{theorem}[Tightness of PSM Procedure for Linear Programs]
\label{thm:lin-tightness}
For the linear program
\[
\text{\emph{(P)}}\qquad
o_{\min}\;=\;
\min_{\mathbf{v}\in\mathbb{R}^n}\;\mathbf{c}^{\top}\mathbf{v}+d
\quad\text{s.t.}\quad A\mathbf{v}\le\mathbf{b},
\]
let \(\mathcal{M}=(\Theta,L)\) be the Parametric Scalar Map produced by the
procedure of Appendix~\ref{app:derivation}.  
Then the optimal bound \(o_{\min}\) is realised by \(\mathcal{M}\); that is,
\[
\exists\,\theta\in\Theta:\;L(\theta)=o_{\min}.
\]
\end{theorem}
\begin{proof}
Assume the primal LP (P) is feasible and has a finite optimum \(o_{\min}\);
otherwise the claim is vacuous.

\textit{Dual form.}  
After separating variable bounds as in Appendix~\ref{app:derivation},
the dual objective is
\[
g(\boldsymbol{\lambda})=
\min_{\mathbf{v}\in[\mathbf{l},\mathbf{u}]}
  (\mathbf{c}+A'^{\top}\boldsymbol{\lambda})^{\!\top}\mathbf{v}
  -\boldsymbol{\lambda}^{\top}\mathbf{b}'+d,
\qquad \boldsymbol{\lambda}\ge 0,
\]
and the PSM construction keeps precisely those multipliers giving finite values:
\[
\Theta=\{\boldsymbol{\lambda}\ge 0\mid g(\boldsymbol{\lambda})>-\infty\},
\quad L(\boldsymbol{\lambda})=g(\boldsymbol{\lambda}).
\]

\textit{Strong duality.}  
Feasibility and boundedness of (P) imply strong duality:
\[
o_{\min}=\max_{\boldsymbol{\lambda}\ge 0}\,g(\boldsymbol{\lambda}).
\]
Hence there exists a dual optimal \(\boldsymbol{\lambda}^*\ge 0\) with
\(g(\boldsymbol{\lambda}^*)=o_{\min}\).

\textit{Membership in \(\Theta\).}  
Since \(g(\boldsymbol{\lambda}^*)=o_{\min}>-\infty\), the multiplier
\(\boldsymbol{\lambda}^*\) satisfies the finiteness condition defining
\(\Theta\), so \(\boldsymbol{\lambda}^*\in\Theta\).

\textit{Conclusion.}  
Setting \(\theta=\boldsymbol{\lambda}^*\) gives
\(L(\theta)=g(\boldsymbol{\lambda}^*)=o_{\min}\),
so the optimal bound is realised by \(\mathcal{M}\).

\medskip
\noindent\textbf{Remark.}
This tightness result hinges on strong duality, which holds for linear
programs but need not hold for general quadratic objectives, especially
non-convex ones where a positive duality gap may exist. In the linear
case, strong duality guarantees that the dual optimum exactly matches
the primal optimum \(o_{\min}\), so \(\mathcal{M}\) always achieves the
tightest possible sound bound. For general quadratic objectives, a
positive duality gap may cause the dual optimum to fall strictly below
\(o_{\min}\), and even this tightness guarantee fails: \(\mathcal{M}\)
may not be able to produce \(o_{\min}\) as a bound at all. Note also
that \(\mathcal{M}\) is currently not complete, in the sense that not
every finite \(\ell \le o_{\min}\) need be realizable: the box-peeled
construction trades completeness over all lower bounds for tractability
of optimization over \(\Theta\). Completeness could in principle be
recovered by dualising the box constraints as well, at the cost of
introducing additional multipliers and making optimization over
\(\Theta\) significantly harder.
\end{proof}

\begin{theorem}[Interval Relaxation Behavior]
\label{thm:interval-start}
The symbolic procedure described in Appendix~\ref{app:derivation} constructs a Parametric Scalar Map $\mathcal{M} = (\Theta, L)$ for optimization problems of the form:
\[
o_{\min} \;=\; 
\min_{\mathbf{v} \in \mathbb{R}^n}\; f(\mathbf{v})
\quad \text{s.t.}\quad A\,\mathbf{v} \le \mathbf{b},
\]
where the objective~$f(\mathbf{v})$ is a quadratic polynomial given by:
\[
f(\mathbf{v}) = 
\sum_{1 \le i < k \le n} H_{ik}\, v_i v_k
+ \sum_{i=1}^{n} Q_i\, v_i^2
+ \mathbf{c}^{\top}\mathbf{v}
+ d.
\]
When the parameters are set to zero, i.e., $\theta = (\boldsymbol{\lambda}, \mathbf{S}, \mathbf{D}) = (\vec{0}, \vec{0}, \vec{0})$, the behavior of the map~$\mathcal{M}$ coincides with that of interval relaxation. If interval relaxation yields a finite lower bound, then $\vec{0} \in \Theta$ and $L(\vec{0})$ matches the interval relaxation value. Otherwise, if the interval relaxation is unbounded below, then $\vec{0} \notin \Theta$.
\end{theorem}

\begin{proof}
As described in Appendix~\ref{app:derivation}, the Parametric Scalar Map $\mathcal{M} = (\Theta, L)$ is constructed using:
\begin{itemize}
    \item Dual multipliers $\boldsymbol{\lambda} \in \mathbb{R}^m$ for the constraints $A\,\mathbf{v} \le \mathbf{b}$,
    \item Splitting parameters $\mathbf{S} \in \mathbb{R}^n$ and $\mathbf{D} \in \mathbb{R}^{n \times n}$ for distributing linear terms.
\end{itemize}

The lower bound is computed by the relaxed expression:
\begin{align*}
\hat g(\boldsymbol{\lambda},\mathbf{S},\mathbf{D}) =\;
&\sum_{1 \le i < k \le n}
  \min_{(v_i,v_k)\in[l_i,u_i]\times[l_k,u_k]}
       \left( H_{ik}v_i v_k + D_{ik}^{i}v_i + D_{ik}^{k}v_k \right) \\
&+\sum_{i=1}^{n}
  \min_{v_i\in[l_i,u_i]}
       \left( Q_i v_i^{2} + S_i v_i \right) \\
&+\min_{\mathbf{v}\in[\mathbf{l},\mathbf{u}]}
       \left( \left(\mathbf{c}+A^{\top}\boldsymbol{\lambda}-\mathbf{S}-\boldsymbol{\Delta}\right)^{\!\top}\mathbf{v} \right)
       + d - \boldsymbol{\lambda}^{\top}\mathbf{b}
\end{align*}

When $(\boldsymbol{\lambda}, \mathbf{S}, \mathbf{D}) = (\vec{0}, \vec{0}, \vec{0})$, this simplifies to:
\begin{align*}
\hat g(\vec{0}, \vec{0}, \vec{0}) =\;
&\sum_{1 \le i < k \le n}
  \min_{(v_i,v_k)\in[l_i,u_i]\times[l_k,u_k]}
        H_{ik}v_i v_k \\
&+\sum_{i=1}^{n}
  \min_{v_i\in[l_i,u_i]}
       Q_i v_i^{2} \\
&+\min_{\mathbf{v}\in[\mathbf{l},\mathbf{u}]}
       \mathbf{c}^{\top}\mathbf{v}
       + d
\end{align*}

This is precisely the interval relaxation of~$f(\mathbf{v})$ over the input box~$[l_i, u_i]$, ignoring the linear constraints~$A\,\mathbf{v} \le \mathbf{b}$. During evaluation, Section~\ref{subsec:inner-parts-eval} adds symbolic feasibility constraints to ensure that any parameter tuple leading to a bound of $-\infty$ is excluded from the parameter space~$\Theta$. In particular, if the interval relaxation yields~$-\infty$, then the zero-point configuration $(\boldsymbol{\lambda}, \mathbf{S}, \mathbf{D}) = (\vec{0}, \vec{0}, \vec{0})$ is ruled out by these constraints, and $\vec{0} \notin \Theta$. Otherwise, if the interval relaxation yields a finite value, then $\vec{0} \in \Theta$, and the corresponding bound~$L(\vec{0})$ exactly matches the interval relaxation value. Thus, the Parametric Scalar Map recovers interval relaxation as the zero-parameter instantiation, when valid.

\end{proof}

\begin{theorem}[Polyhedrality of the feasible‐parameter set]
\label{thm:theta-polyhedral}
Let 
\(
\theta=(\boldsymbol{\lambda},\mathbf S,\mathbf D)\in
\mathbb{R}^{m+n+n^{2}}
\)
be the parameter vector constructed by the PSM computing procedure in
Appendix~\ref{app:derivation}.  
The set of admissible parameters
\[
\Theta \;=\;
\bigl\{\theta\mid\text{all feasibility checks introduced
in Steps\,2–4 are satisfied}\bigr\}
\]
is a (possibly empty) \emph{polyhedron}; equivalently, there exist a matrix
\(M\) and vector \(\mathbf h\) such that
\(
\Theta=\{\theta\mid M\theta\le\mathbf h\}.
\)
\end{theorem}

\begin{proof}
We enumerate the constraints added during the construction and show that each
is linear in~\(\theta\).

\smallskip
\noindent\textbf{(i)~Dual non-negativity.}  
Step 2 fixes \(\boldsymbol{\lambda}\ge\mathbf 0\), that is,
\(\lambda_i\ge 0\) for \(i=1,\dots,m\); these are linear half-spaces.

\smallskip
\noindent\textbf{(ii)~Bilinear sub-problems.}  
For every pair \((i,k)\) the rule in
Section~\ref{subsec:inner-parts-eval} may add
\[
u_2\le 0,\;u_1\le 0,\;l_2\ge 0,\text{ or }l_1\ge 0
\quad\text{or the contradictory }0\le -1.
\]
Here \(l_1,u_1,l_2,u_2\) are affine expressions in
the splitting parameters \(D_{ik}^{i},D_{ik}^{k}\); hence each inequality is
linear in~\(\theta\).

\smallskip
\noindent\textbf{(iii)~Quadratic sub-problems.}  
When \(A<0\) and a bound is missing, the procedure inserts the contradictory
half–space \(0\le -1\); otherwise no new restriction is
added.  Again, \(0\le -1\) is linear.

\smallskip
\noindent\textbf{(iv)~Linear sub-problem.}  
For every coordinate \(i\) the guard may require
\(c_i\le 0,\;c_i\ge 0,\) or \(c_i=0\).
Because each coefficient takes the affine form
\(c_i=\hat c_i(\boldsymbol{\lambda},\mathbf S,\mathbf D)\),
these relations are linear in~\(\theta\).

\smallskip
\noindent\textbf{(v)~Closure under intersection.}  
The parameter set \(\Theta\) is the finite intersection of the half-spaces
listed in (i)–(iv).  A finite intersection of linear half-spaces is a
polyhedron by definition.

Hence \(\Theta\) is polyhedral.
\end{proof}


\begin{theorem}[Differentiability of L function in PSM]
\label{thm:L-differentiable}
Let 
\(
\theta = (\boldsymbol{\lambda}, \mathbf{S}, \mathbf{D}) \in \mathbb{R}^{m+n+n^2}
\)
be the parameter vector constructed by the PSM computing procedure in Appendix~\ref{app:derivation}. This procedure outputs a parametric sound map (PSM) 
\(
\mathcal{M} = (\Theta, L)
\),
where \( \Theta \subseteq \mathbb{R}^{m+n+n^2} \) is the polyhedral set of admissible parameters (as established in Theorem~\ref{thm:theta-polyhedral}) and 
\(
L : \Theta \to \mathbb{R}
\)
is the function that maps each admissible \( \theta \in \Theta \) to the minimum value of the relaxed dual objective \( \hat g(\theta) \). Then the function \( L(\theta) \) is piecewise-defined and differentiable on the interior of \( \Theta \).

More precisely:
\begin{itemize}
    \item On any region of \( \Theta \) where the minimizing terms and endpoints in all subproblems (bilinear, quadratic, and linear) remain fixed, the expression for \( L(\theta) \) is a differentiable function of~\( \theta \).
    \item The function \( L(\theta) \) is continuous on all of \( \Theta \), and differentiable almost everywhere in \( \Theta \), except possibly at boundaries where the active minimizer (e.g., which endpoint of an interval achieves the minimum) changes.
\end{itemize}
\end{theorem}

\begin{proof}
The function $L(\theta)$ is computed as described in the derivation procedure in Appendix~\ref{app:derivation}, by summing the minimum values of several subproblems, one for each bilinear, quadratic, and linear term in the relaxed objective. We analyze the differentiability of each class of subproblem:

\smallskip
\noindent\textbf{(i)~Linear terms.} The minimized value of each $c_i v_i$ term is computed based on whether $c_i$ is positive, negative, or zero. The result is a piecewise-linear function of $c_i$, and since $c_i = \hat c_i(\boldsymbol{\lambda}, \mathbf{S}, \mathbf{D})$ is affine in $\theta$, the contribution $L_i$ is piecewise-affine in $\theta$, hence differentiable except at switching points (where the sign of $c_i$ changes).

\smallskip
\smallskip
\noindent\textbf{(ii)~Quadratic terms.} Consider subproblems of the form $\min_{v_i \in [l_i, u_i]} A v_i^2 + B v_i$ with fixed $A \ne 0$ and parameter-dependent $B = \hat{B}(\theta)$. The analysis splits into two cases:

\begin{itemize}
\item If $A > 0$ (convex), the objective is minimized at the unconstrained critical point $v_i^* = -B/(2A)$. If $v_i^* \in [l_i, u_i]$, the minimum is attained there and equals $A (v_i^*)^2 + B v_i^*$, which is a smooth function of $B$, hence of $\theta$. If $v_i^* \notin [l_i, u_i]$, the minimum occurs at one of the endpoints. Since $B$ is affine in $\theta$, both $A l_i^2 + B l_i$ and $A u_i^2 + B u_i$ are smooth in $\theta$, and their minimum is piecewise-smooth depending on which endpoint is active. Thus, in both cases, the quadratic contribution is differentiable on regions where the active minimizer remains fixed.

\item If $A < 0$ (concave), the objective can decrease without bound unless both $l_i$ and $u_i$ are finite. In particular, the term $A v_i^2$ dominates and diverges to $-\infty$ if $v_i \to \pm \infty$. To ensure boundedness, the feasibility check adds the constraint that both bounds must be finite. If so, the minimum is computed as $\min\{A l_i^2 + B l_i,\; A u_i^2 + B u_i\}$, which is continuous and piecewise-differentiable in $\theta$.
\end{itemize}

\smallskip
\noindent\textbf{(iii)~Bilinear terms.} Consider the subproblem
\[
\min_{(v_i, v_k) \in [l_i, u_i] \times [l_k, u_k]} \Big(A\, v_i v_k + B\, v_i + C\, v_k\Big)\,,
\]
with $A \ne 0$ fixed and $B$, $C$ affine in $\theta$. Rewriting:
\[
A v_i v_k + B v_i + C v_k = (A v_i + C)(v_k + B/A) - \frac{B C}{A}\,,
\]
we reduce to minimizing a product of two affine expressions over intervals $[l_1, u_1]$ and $[l_2, u_2]$, where:
\[
l_1 := \min(A l_i + C,\; A u_i + C), \quad l_2 := l_k + B/A\,, \quad \text{etc.}
\]

The minimal value is:
\[
l_{\text{bilin}} := \min\{l_1 l_2,\; l_1 u_2,\; u_1 l_2,\; u_1 u_2\} - \frac{B C}{A}\,.
\]

Unboundedness arises in four cases (e.g., $l_1 = -\infty$ and $u_2 > 0$). If the paired bound is also infinite, we add the contradictory constraint $0 \le -1$ to $\Theta$; otherwise, we insert the appropriate bound constraint (e.g., $u_2 \le 0$). This ensures the subproblem remains well-posed and the resulting $l_{\text{bilin}}$ is finite and piecewise-smooth in $\theta$.

\smallskip
\noindent\textbf{(iv)~Summation.} The total value $L(\theta)$ is the sum of these subproblem contributions. Since each term is continuous and piecewise-differentiable, their sum is also continuous and piecewise-differentiable. Differentiability holds on any region where the active minimizers in all subproblems are fixed.
\end{proof}

\begin{theorem}[EUM over-approximates a sequence]\label{thm:eum-over}
Let $\tilde{S}$ be a sequence of assignments, and let $\block(\tilde{S})$ have concrete semantics given by its EUM. Then
\[
\llbracket \tilde{S} \rrbracket \;\subseteq\; \llbracket \block(\tilde{S}) \rrbracket .
\]
\end{theorem}
\begin{proof}
Let $(S,T)$ be program states (valuations) with $(S,T)\in \llbracket \tilde{S} \rrbracket$, and write
$\tilde{S}$ as $x_1:=e_1;\,x_2:=e_2;\,\dots;\,x_k:=e_k$.
Unfolding the run gives intermediate states $U_0=S,U_1,\dots,U_k=T$, where step $i$ sets
$x_i$ to the value of $e_i$ in $U_{i-1}$ and leaves other variables unchanged.
Thus $T(x_1)$ depends on $S$, $T(x_2)$ depends on $S$ and the new value of $x_1$, and in general
$T(x_i)$ depends on $S$ and the updated values $x_1,\dots,x_{i-1}$.
By a straightforward induction on $i$, there exist flattened expressions $F_v$ such that
\[
T(v)=\text{value of }F_v\text{ in }S\quad\text{for all }v,
\]
i.e., all intermediate states $U_1,\dots,U_{k-1}$ can be eliminated by substitution so that each
final variable is expressed purely in terms of the initial state.
The EUM for the block captures exactly this flattened relation: it is obtained from these equations
by symbolic rewriting and algebraic simplification, in a way that it exactly computes the function $F$, and so $(S,T)\in \llbracket \block(\tilde{S}) \rrbracket$. Since $(S,T)$ was arbitrary, we conclude
$\llbracket \tilde{S} \rrbracket \subseteq \llbracket \block(\tilde{S}) \rrbracket$.
\paragraph{Why inclusion can be strict.}
Symbolic simplification in the EUM can hide intermediate errors and thereby admit extra states.
\emph{Example.} Consider
\[
a := a_1;\quad b := a_1;\quad c := a/d - b/d .
\]
Flattening gives \(a'=a_1,\ b'=a_1,\ c'=(a_1/d)-(a_1/d)=0\).  
If the EUM keeps only these simplified postconditions and omits the condition on the intermediate divisions, then states with \(d=0\) and \(c'=0\) satisfy the EUM even though the stepwise execution would fail on the divisions. Thus new states are added and \(\llbracket \tilde{S} \rrbracket \subsetneq \llbracket \block(\tilde{S}) \rrbracket\).
\end{proof}

\noindent \textbf{Remark.} The sequences $\seq$ considered by \toolname are those whose combined effect lies within the QGO class. Such sequences consist solely of affine and quadratic assignments, which are total operations defined for all input states. For such sequences, no intermediate errors can arise during symbolic flattening, and the EUM captures the exact semantics without introducing additional states, i.e., $\llbracket \seq \rrbracket = \llbracket \block(\seq) \rrbracket$.

\clearpage

\section{Evaluation Data}

\subsection{Linear Case}
\label{sec:appeval1}

\begin{table}[H]
\caption{Full stats for Zones Experiments (Affine)}
\centering
\setlength{\tabcolsep}{5pt}
\small
\begin{tabular}{@{}ccccc@{}}
\toprule
\makecell{\textbf{Analysis}\\ \textbf{Method}} & \makecell{\textbf{Runtime}\\ \textbf{(s)}} & \makecell{\textbf{Invariants Strengthened}\\ \textbf{over ELINA (OG)}} & \makecell{\textbf{Invariants Strengthened}\\ \textbf{from Prev. Step (MG)}} & \makecell{\textbf{Remaining Gap (RG)}\\ \textbf{from most-precise}} \\
\midrule
ELINA & 18.53 & 0 & - & 121 \\
AbsEvolve-R-0 & 60.47 & 91 & 91 & 48 \\
AbsEvolve-R-1 & 61.82 & 97 & 10 & 42 \\
AbsEvolve-R-2 & 65.2 & 116 & 21 & 30 \\
AbsEvolve-R-3 & 68.8 & 121 & 30 & 4 \\
AbsEvolve-R-4 & 69.59 & 121 & 0 & 4 \\
AbsEvolve-R-5 & 71.9 & 121 & 4 & 0 \\
LP Solver Based & 230.74 & 121 & - & 0 \\
\bottomrule
\end{tabular}
\label{table:zones-lin-full}
\end{table}

\begin{table}[H]
\caption{Full stats for Octagon Experiments (Affine)}
\centering
\setlength{\tabcolsep}{5pt}
\small
\begin{tabular}{@{}ccccc@{}}
\toprule
\makecell{\textbf{Analysis}\\ \textbf{Method}} & \makecell{\textbf{Runtime}\\ \textbf{(s)}} & \makecell{\textbf{Invariants Strengthened}\\ \textbf{over ELINA (OG)}} & \makecell{\textbf{Invariants Strengthened}\\ \textbf{from Prev. Step (MG)}} & \makecell{\textbf{Remaining Gap (RG)}\\ \textbf{from most-precise}} \\
\midrule
ELINA & 20.11 & 0 & - & 59 \\
AbsEvolve-R-0 & 82.07 & 24 & 24 & 55 \\
AbsEvolve-R-1 & 83.9 & 42 & 18 & 41 \\
AbsEvolve-R-2 & 89.07 & 57 & 31 & 29 \\
AbsEvolve-R-3 & 102.66 & 59 & 29 & 4 \\
AbsEvolve-R-4 & 108.68 & 59 & 0 & 4 \\
AbsEvolve-R-5 & 111.06 & 59 & 4 & 0 \\
LP Solver Based & 344.01 & 59 & - & 0 \\
\bottomrule
\end{tabular}
\label{table:oct-lin-full}
\end{table}

\subsection{Nonlinear Operators and Expressive Domains}
\label{sec:appeval2}

The following programs give numerical overflow errors while running polyhedra analysis with ELINA: geo1-u.c, geo1-u2.c, fermat1.c, fermat2.c, fermat1-ll.c, fermat2-ll.c, dijkstra.c and hard.c

\begin{table}[H]
\caption{Full stats for Zones Experiments (Affine + Quadratic)}
\centering
\setlength{\tabcolsep}{5pt}
\small
\begin{tabular}{@{}cccc@{}}
\toprule
\makecell{\textbf{Analysis}\\ \textbf{Method}} & \makecell{\textbf{Runtime}\\ \textbf{(s)}} & \makecell{\textbf{Invariants Strengthened}\\ \textbf{over ELINA (OG)}} & \makecell{\textbf{Invariants Strengthened}\\ \textbf{from Prev. Step (MG)}} \\
\midrule
ELINA & 15.48 & 0 & - \\
AbsEvolve-R-0 & 68.09 & 110 & 110 \\
AbsEvolve-R-1 & 75.34 & 112 & 6 \\
AbsEvolve-R-2 & 79.77 & 131 & 19 \\
AbsEvolve-R-3 & 85.86 & 135 & 25 \\
AbsEvolve-R-4 & 95.78 & 135 & 0 \\
AbsEvolve-R-5 & 97.94 & 135 & 4 \\
\bottomrule
\end{tabular}
\label{table:zones-quad-full}
\end{table}

\begin{table}[H]
\caption{Full stats for Octagon Experiments (Affine + Quadratic)}
\centering
\setlength{\tabcolsep}{5pt}
\small
\begin{tabular}{@{}cccc@{}}
\toprule
\makecell{\textbf{Analysis}\\ \textbf{Method}} & \makecell{\textbf{Runtime}\\ \textbf{(s)}} & \makecell{\textbf{Invariants Strengthened}\\ \textbf{over ELINA (OG)}} & \makecell{\textbf{Invariants Strengthened}\\ \textbf{from Prev. Step (MG)}} \\
\midrule
ELINA & 16.44 & 0 & - \\
AbsEvolve-R-0 & 118.39 & 95 & 95 \\
AbsEvolve-R-1 & 125.75 & 97 & 14 \\
AbsEvolve-R-2 & 137.52 & 111 & 40 \\
AbsEvolve-R-3 & 154.76 & 113 & 25 \\
AbsEvolve-R-4 & 158.3 & 113 & 0 \\
AbsEvolve-R-5 & 168.08 & 113 & 14 \\
\bottomrule
\end{tabular}
\label{table:oct-quad-full}
\end{table}

\begin{table}[H]
\caption{Full stats for Polyhedra Experiments (Quadratic)}
\centering
\setlength{\tabcolsep}{5pt}
\small
\begin{tabular}{@{}cccc@{}}
\toprule
\makecell{\textbf{Analysis}\\ \textbf{Method}} & \makecell{\textbf{Runtime}\\ \textbf{(s)}} & \makecell{\textbf{Invariants Strengthened}\\ \textbf{over ELINA (OG)}} & \makecell{\textbf{Invariants Strengthened}\\ \textbf{from Prev. Step (MG)}} \\
\midrule
ELINA & 23.2 & 0 & - \\
AbsEvolve-R-0 & 99.95 & 38 & 38 \\
AbsEvolve-R-1 & 107.1 & 38 & 0 \\
AbsEvolve-R-2 & 112.88 & 40 & 40 \\
AbsEvolve-R-3 & 126.95 & 43 & 9 \\
AbsEvolve-R-4 & 130.42 & 45 & 3 \\
AbsEvolve-R-5 & 134.16 & 45 & 0 \\
\bottomrule
\end{tabular}
\label{table:pk-quad-full}
\end{table}

\subsection{Results without analyzing sequences together}
\label{sec:appeval3}

\begin{figure}[H]
    \centering
    \begin{subfigure}[b]{0.32\linewidth}
        \centering
        \includegraphics[width=\linewidth]{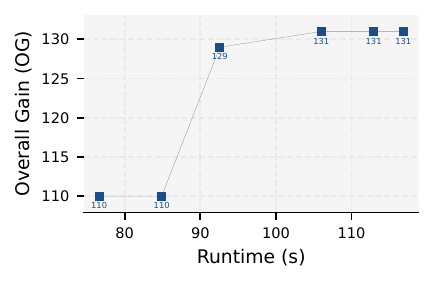}
        \vspace{0.25em}
        \includegraphics[width=\linewidth]{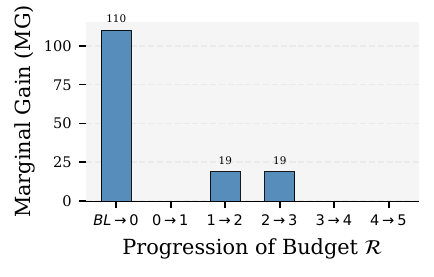}
        \caption{Zones}
    \end{subfigure}
    \hfill
    \begin{subfigure}[b]{0.32\linewidth}
        \centering
        \includegraphics[width=\linewidth]{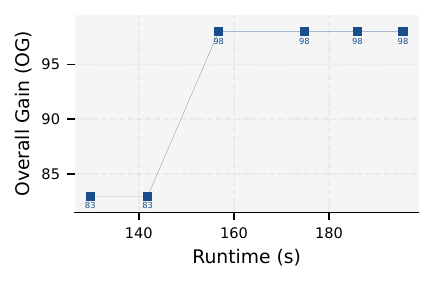}
        \vspace{0.25em}
        \includegraphics[width=\linewidth]{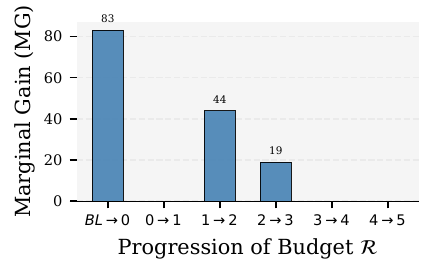}
        \caption{Octagons}
    \end{subfigure}
    \hfill
    \begin{subfigure}[b]{0.32\linewidth}
        \centering
        \includegraphics[width=\linewidth]{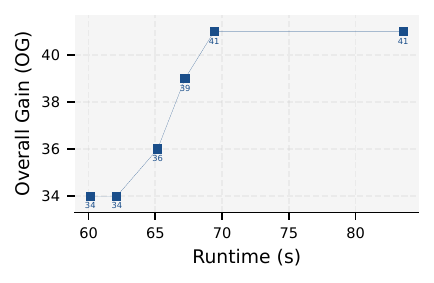}
        \vspace{0.25em}
        \includegraphics[width=\linewidth]{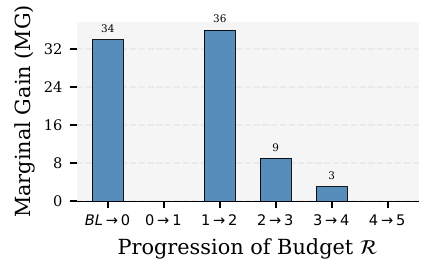}
        \caption{Polyhedra}
    \end{subfigure}

    \vspace{0.4em}
    \caption{Evolution of invariant strengthening while handling quadratic assignments across domains.}
    \vspace{-1em}
\end{figure}

\begin{table}[H]
\caption{Stats for Zones Experiments without sequence merging (Affine + Quadratic)}
\centering
\setlength{\tabcolsep}{5pt}
\small
\begin{tabular}{@{}cccc@{}}
\toprule
\makecell{\textbf{Analysis}\\ \textbf{Method}} & \makecell{\textbf{Runtime}\\ \textbf{(s)}} & \makecell{\textbf{Invariants Strengthened}\\ \textbf{over ELINA (OG)}} & \makecell{\textbf{Invariants Strengthened}\\ \textbf{from Prev. Step (MG)}} \\
\midrule
ELINA & 13.53 & 0 & - \\
AbsEvolve-R-0 & 76.62 & 110 & 110 \\
AbsEvolve-R-1 & 84.89 & 110 & 0 \\
AbsEvolve-R-2 & 92.58 & 129 & 19 \\
AbsEvolve-R-3 & 105.99 & 131 & 19 \\
AbsEvolve-R-4 & 112.91 & 131 & 0 \\
AbsEvolve-R-5 & 116.84 & 131 & 0 \\
\bottomrule
\end{tabular}
\label{table:zones-quad-nc}
\end{table}

\begin{table}[H]
\caption{Stats for Octagon Experiments without sequence merging (Affine + Quadratic)}
\centering
\setlength{\tabcolsep}{5pt}
\small
\begin{tabular}{@{}cccc@{}}
\toprule
\makecell{\textbf{Analysis}\\ \textbf{Method}} & \makecell{\textbf{Runtime}\\ \textbf{(s)}} & \makecell{\textbf{Invariants Strengthened}\\ \textbf{over ELINA (OG)}} & \makecell{\textbf{Invariants Strengthened}\\ \textbf{from Prev. Step (MG)}} \\
\midrule
ELINA & 15.08 & 0 & - \\
AbsEvolve-R-0 & 129.81 & 83 & 83 \\
AbsEvolve-R-1 & 141.89 & 83 & 0 \\
AbsEvolve-R-2 & 156.8 & 98 & 44 \\
AbsEvolve-R-3 & 174.81 & 98 & 19 \\
AbsEvolve-R-4 & 185.94 & 98 & 0 \\
AbsEvolve-R-5 & 195.54 & 98 & 0 \\
\bottomrule
\end{tabular}
\label{table:oct-quad-nc}
\end{table}

\begin{table}[H]
\caption{Stats for Polyhedra Experiments without sequence merging (Quadratic)}
\centering
\setlength{\tabcolsep}{5pt}
\small
\begin{tabular}{@{}cccc@{}}
\toprule
\makecell{\textbf{Analysis}\\ \textbf{Method}} & \makecell{\textbf{Runtime}\\ \textbf{(s)}} & \makecell{\textbf{Invariants Strengthened}\\ \textbf{over ELINA (OG)}} & \makecell{\textbf{Invariants Strengthened}\\ \textbf{from Prev. Step (MG)}} \\
\midrule
ELINA & 26.93 & 0 & - \\
AbsEvolve-R-0 & 60.13 & 34 & 34 \\
AbsEvolve-R-1 & 62.12 & 34 & 0 \\
AbsEvolve-R-2 & 65.17 & 36 & 36 \\
AbsEvolve-R-3 & 67.23 & 39 & 9 \\
AbsEvolve-R-4 & 69.41 & 41 & 3 \\
AbsEvolve-R-5 & 83.57 & 41 & 0 \\
\bottomrule
\end{tabular}
\label{table:pk-quad-nc}
\end{table}
\section{Extension to Joins and Disjunctions}
\label{app:joins}

The UPOSE algorithm extends naturally to joins and disjunctions. Consider 
a disjunctive feasible set $S_1 \cup S_2$, arising for instance from a 
disjunctive guard such as $x \leq y \lor z \leq a$. For template direction 
$i$, the bound to compute is:
$$
m_i := \min_{x \in S_1 \cup S_2} f_i(x) = \min(m_{i,1}, m_{i,2}), 
\quad \text{where} \quad m_{i,j} := \min_{x \in S_j} f_i(x).
$$
For each branch $j \in \{1, 2\}$, UPOSE constructs a parametric lower bound 
$L_{i,j}(\lambda_j)$ over a polyhedral parameter space $P_{i,j}$, such that 
$\forall \lambda_j \in P_{i,j}, \; L_{i,j}(\lambda_j) \leq m_{i,j}$. The 
two branches are then combined as:
$$
L_i(\lambda_1, \lambda_2) := \min(L_{i,1}(\lambda_1), L_{i,2}(\lambda_2)),
$$
over the joint polyhedral parameter space $P_i = \{(\lambda_1, \lambda_2) \mid 
\lambda_1 \in P_{i,1}, \lambda_2 \in P_{i,2}\}$. This yields a polyhedral 
parameter space with a piecewise differentiable objective, on which AGG can 
be applied directly. More broadly, this demonstrates that our approach is 
general and can naturally accommodate joins and disjunctions, and once 
supported, they immediately work across all domains handled by the framework, 
including new user-defined TCM domains.

\end{document}